\documentclass[a4paper,UKenglish]{lipics}
\usepackage{microtype}%if unwanted, comment out or use option "draft"
\usepackage{tikz}
\usetikzlibrary{calc,shapes.geometric,positioning}
\usepackage{listings}
\usepackage{subfigure}
\usepackage{floatflt}
\usepackage{wrapfig}

\title{A New Upper Bound for the Traveling Salesman Problem in Cubic Graphs\footnote{This paper has been  formatted using the LIPIcs LaTeX template.}}
\titlerunning{New Upper Bound for TSP in Cubic Graphs} %optional, in case that the title is too long; the running title should fit into the top page column

\author[1]{Maciej Li\'{s}kiewicz}
\author[1]{Martin  R. Schuster}
\affil[1]{Institute of Theoretical Computer Science, University of Lübeck\\
          Ratzeburger Allee 160, 23538 Lübeck, Germany \\
         \texttt{liskiewi@tcs.uni-luebeck.de, schuster@tcs.uni-luebeck.de}}
\authorrunning{M. Li\'{s}kiewicz, M. R. Schuster} %mandatory. First: Use abbreviated first/middle names. Second (only in severe cases): Use first author plus 'et. al.'

% mandatory: Please choose ACM 1998 classifications from http://www.acm.org/about/class/ccs98-html . E.g., cite as "F.1.1 Models of Computation". 
% mandatory: Please provide 1-5 keywords
\newcommand{\caO}{{\cal O}}

\begin{document}

\maketitle

\begin{abstract}
We provide a new upper bound for traveling salesman problem (TSP) in cubic graphs, i.\,e.\ graphs with maximum vertex degree three, and prove that the problem for an $n$-vertex graph can be solved in $\mathcal{O}(1.2553^n)$ time and in linear space. We show that the exact TSP algorithm of Eppstein, with some minor modifications, yields the stated result. The previous best known upper bound $\mathcal{O}(1.251^n)$ was claimed by Iwama and Nakashima [Proc. COCOON 2007]. Unfortunately, their analysis contains several mistakes that render the proof for the upper bound invalid.
\end{abstract}

\lstset{language=Matlab,numbers=left,numberstyle=\footnotesize,showspaces=false,showstringspaces=false,breaklines=true}

\tikzstyle{vertex}=[circle,fill=black,minimum size=3pt,inner sep=0pt]
\tikzstyle{edge} = [draw,-]
\tikzstyle{edgef} = [draw,line width=1.5pt,-]
\tikzstyle{edgesel} = [draw,color=black!30,line width=1.5pt,-] %[draw,line width=1.5pt,dotted,-]
\tikzstyle{edgerem} = [draw,dotted,-]
\tikzstyle{edgechoose} = [draw,dashed,-]

\section{Introduction}

It is an outstanding open problem whether the traveling salesman problem (TSP) and the closely related Hamiltonian cycle problem can be solved in $\caO(c^n)$ time for graphs on $n$ vertices, for some constant $c<2$. Recently Björklund et~al. \cite{BjorklundHKK12} have shown that the classical Bellman-Held-Karp exact algorithms \cite{Bellman62,HeldK62} for solving TSP can be modified to run in time $\caO((2 - \varepsilon)^n)$, where $\varepsilon > 0$ depends only on the maximum vertex degree. This provides the first upper bound on the time complexity of TSP that lies below $2^n$ for a broad class of graphs such as bounded degree graphs. Particularly, applying the result of \cite{BjorklundHKK12} for graphs with maximum vertex degree three, also called cubic graphs, one gets that TSP can be solved in time $2^{3n/4}n^{\caO(1)} = \caO(1.682^n)$. On the other hand, the problem of testing whether a cubic graph has a Hamiltonian cycle and consequently the decision version of TSP remain NP-complete even if the graphs 
are restricted to be planar \cite{GareyJT76}.

Exact algorithms for TSP for special classes of bounded degree graphs, in particular for cubic graphs, have been the subject of separate studies. The motivation for the study comes both from theoretical concerns and from practical applications, e.\,g.\ in computer graphics \cite{ArkinEtAl96,EppsteinGopi04}. The first exact algorithm for TSP in cubic graphs running faster than in time $2^n$ was proposed by Eppstein \cite{Eppstein07}. His algorithm solves the problem in $2^{n/3}n^{\caO(1)} = \caO(1.260^n)$ time and linear space and additionally it is easy to implement. Thus, although the technique by Björklund et al.\ improves the upper bound  $2^n$ for any degree bounded $\ge 3$, for specific bounds, like e.\,g.\ $3$, better methods exist. 

Eppstein's algorithm is a sophisticated recursive branch-and-bound search, which takes advantage of small vertex degrees in a graph. To speed-up searching, the algorithm uses the fact that in cubic graphs a selection of an edge to a recursively constructed Hamiltonian cycle forces several further edges to be in the cycle or not. In \cite{IwamaN07} Iwama and Nakashima slightly modify Eppstein's algorithm and provide a new interesting method to bound the number of worst-case branches in any path of the branching tree corresponding to recursive subdivisions of the problem. As a consequence, Iwama and Nakashima claim $\caO(1.251^n)$ to be an upper bound for the run-time of the algorithm. Unfortunately, their paper contains several serious mistakes that render the proof for the upper bound invalid (for details, see Section~\ref{section:comments:on:IN}). After reformulating the key lemma of \cite{IwamaN07} to be correct and then using the lemma to solve the recurrences derived in \cite{IwamaN07} in a proper way 
one could prove the upper bound $\caO(1.257^n)$ that still beats the bound  $\caO(1.260^n)$ by Eppstein.%\\[2mm]
%
%\noindent{\bf Our result.} 

\subsection*{Our Result} 
In this article we provide a new upper bound for TSP in cubic graphs. We show that Eppstein's algorithm with some minor modifications, similar to those used in \cite{IwamaN07}, yields the stated result: 
\begin{theorem}\label{theorem:main}
The traveling salesman problem for an $n$-vertex cubic graph can be solved in $\mathcal{O}(1.2553^n)$ time and in linear space.
\end{theorem}
Our proof techniques are based on ideas used by Eppstein \cite{Eppstein07} and Iwama and Nakashima \cite{IwamaN07} and a new, more careful study of worst-case branches in the tree of recursive subdivisions of the problem performed by the algorithm. Thus, our main contribution is more analytical than algorithmic. Nevertheless, we have implemented %the modified 
our algorithm 
and verified its easy implementability and good performance for graphs up to 114 vertices (for the experimental results see~\cite{Schuster2012}). %\\[2mm] 
%
%\noindent{\bf Related work.} 

\subsection*{Related Work} 
Applying the result by Björklund et al.\ for an $n$-vertex graph with maximum degree four one gets that TSP can be solved in  $\caO(1.856^n)$ time and exponential space. Eppstein \cite{Eppstein07} showed that the problem  can be solved in $\caO(1.890^n)$ time but using only polynomial space. Next, Gebauer \cite{Gebauer11} proposed an algorithm that runs in time  $\caO(1.733^n)$. This algorithm can also list the Hamiltonian cycles but it uses exponential space. Very recently, Cygan et al.\  \cite{CyganNPPRW11} have shown a Monte Carlo algorithm with constant one-sided error probability that solves the Hamiltonian cycle problem in $\caO(1.201^n)$ time for cubic graphs and in $\caO(1.588^n)$ time for graphs of maximum degree four. Though the technique presented in \cite{CyganNPPRW11} works well for the Hamiltonian cycle problem, it is not usable for~TSP.%\\[3mm]
%

%\noindent{\bf Paper organization.} 
\subsection*{Paper organization} 
In Section~\ref{secalgorithm} we recall Eppstein's algorithm and in Section~\ref{mainidea} main ideas of our proof techniques are sketched. In Section~\ref{section:our:analysis} we specify our modifications of Eppstein's algorithm and provide its analysis. Section~\ref{sec:proofs} presents the proof of our main technical result on the number of worst-case branches  of the branching tree. Finally, Section~\ref{section:comments:on:IN} discusses some issues concerning the paper by Iwama and Nakashima and the last section presents our conclusions.

\section{Outline of Eppstein's Algorithm}
\label{secalgorithm}
\tikzstyle{smallvertex}=[circle,fill=black,minimum size=2pt,inner sep=0pt]
\tikzstyle{smalledge} = [draw,-]
\tikzstyle{smalledgef} = [draw,line width=1.1pt,-]

Eppstein's TSP algorithm \cite{Eppstein07} (see also Appendix) searches {\em recursively} for a minimum weight Hamiltonian cycle $H_\text{\it min}$ in a given cubic graph $G$. The algorithm constructs successively Hamiltonian cycles which are determined by a set of {\em forced edges} $F$. The goal is to find the set $F$ which coincides with $H_\text{\it min}$. In each recursion step, an edge $e \in G\setminus F$ is chosen (Step~3). Obviously, $e$ either belongs to $H_\text{\it min}$ or not. Thus, the algorithm makes two recursive calls: once $e$ is added to $F$ (Step~4), assuming $e\in H_\text{\it min}$, and once $e$ is removed from $G$ (Step~5), assuming $e\not\in H_\text{\it min}$. The better solution to these two subproblems will then be returned (Step~6).

At the beginning of each recursive call, $G$ and $F$ are simplified iteratively (Step~1). If $F$ contains two edges meeting at a single vertex $v$, the third edge incident to $v$ cannot be part of $H_\text{\it min}$, and therefore it is removed from $G$. If $G$ contains a vertex with degree two, both incident edges have to be in $H_\text{\it min}$, thus, they are added to $F$. If $G$ contains parallel edges, one of the edges is removed from $G$ depending on the weight and on $F$. Next, if $G$ contains a cycle of four unforced edges, two opposite vertices of which are each incident to a forced edge outside the cycle, like here 
\begin{tikzpicture}[node distance=0.1, scale=0.3]
      \node[smallvertex] (a) at (0,0) {};
      \node[smallvertex] (b) [right=of a] {};
      \node[smallvertex] (c) [below=of b] {};
      \node[smallvertex] (d) [left=of c] {};
      \node[smallvertex] (e) [above left=of a] {};
      \node[smallvertex] (f) [above right=of b] {};
      \node[smallvertex] (g) [below right=of c] {};
      \node[smallvertex] (h) [below left=of d] {};
      \foreach \source/ \dest in {a/e, c/g}
           \path[smalledgef] (\source) -- (\dest);
     \foreach \source/ \dest in {a/b, b/c, c/d, d/a, b/f, d/h}
           \path[smalledge] (\source) -- (\dest);
    \end{tikzpicture}\ , 
then all non-cycle edges that are incident to a vertex of the cycle are added to $F$. Furthermore, triangles are contracted to a single vertex by adjusting the weights of attached edges. These techniques reduce the size of the recurrence tree of the algorithm and enforce $G$ to be simple, cubic and triangle-free. Simultaneously, it is verified whether the current set of forced edges $F$ leads to a contradiction: e.\,g.\ $F$ might contain three edges meeting at a single vertex or $F$ contains a non-Hamiltonian cycle. Then the algorithm returns a {\em None-value}.
%A contradiction is also found if $G$ contains a vertex with degree zero or one. 

Another idea to reduce the size of the recurrence tree is to stop the recursion when $G\setminus F$ consists of a set of non-connected 4-cycles. Then a solution is found %(in polynomial time) 
by constructing a minimum spanning tree on some helper graph $G'$ (Step~2).% This way 4-cycles with four attached edges in $F$ have not to be processed recursively.

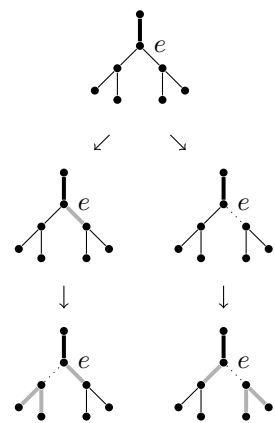
\begin{wrapfigure}{r}{0pt}
% \begin{figure}[htbp]
  \centering
%   \subfigure[Branching, no Step~1]{
%     \hspace{0.2cm}
    \begin{tikzpicture}[node distance=0.3*1, scale=1]
      \node[vertex] (a1) at (1,2.1) {};
      \node[vertex] (a2) at (0,0) {};
      \node[vertex] (a3) at (2.1,0) {};
      \node[vertex] (a4) at (0,-2.1) {};
      \node[vertex] (a5) at (2.1,-2.1) {};

      \draw [->] (0.6,0.5) -- (0.4,0.3);
      \draw [->] (1.4,0.5) -- (1.6,0.3);

      \draw [->] (0,-1.5) -- (0,-1.8);
      \draw [->] (2.1,-1.5) -- (2.1,-1.8);

      \foreach \ind in {1,2,3,4,5} {
        \node[vertex,label=right:$e$] (y\ind) [below=of a\ind] {};
        \node[vertex] (x\ind) [below left=of y\ind] {};
        \node[vertex] (z\ind) [below right=of y\ind] {};
        \node[vertex] (b\ind) [below left=of x\ind] {};
        \node[vertex] (c\ind) [below right=of z\ind] {};
        \node[vertex] (d\ind) [below=of z\ind] {};
        \node[vertex] (e\ind) [below=of x\ind] {};
        \foreach \source/ \dest in {a\ind/y\ind}
          \path[edgef] (\source) -- (\dest);
      }

      \foreach \source/ \dest in {x1/y1, y1/z1, b1/x1, e1/x1, c1/z1, d1/z1, c2/z2, d2/z2, b3/x3, e3/x3, b2/x2, e2/x2, x3/y3, c3/z3, d3/z3, x2/y2}
        \path[edge] (\source) -- (\dest);
      \foreach \source/ \dest in {y2/z2}
        \path[edgesel] (\source) -- (\dest);
      \foreach \source/ \dest in {y3/z3}
        \path[edgerem] (\source) -- (\dest);

      \foreach \source/ \dest in {c4/z4, d4/z4, b5/x5, e5/x5}
        \path[edge] (\source) -- (\dest);
      \foreach \source/ \dest in {b4/x4, e4/x4, y4/z4, x5/y5, c5/z5, d5/z5}
        \path[edgesel] (\source) -- (\dest);
      \foreach \source/ \dest in {x4/y4, y5/z5}
        \path[edgerem] (\source) -- (\dest);
    \end{tikzpicture}
    \caption{A branch and subsequent simplifications of $G$~and~$F$.}
    \label{fig:branching}
%   \caption{Branching; bold black line: edge in $F$; bold gray line: edge selected by the current branch or forced to be selected in Step~1 directly after the branch; dotted line: edge removed by the branch or forced to not appear in a cycle.}
\end{wrapfigure}
The choice of an edge $e$ (Step~3) which specifies two subproblems for the recursive calls, plays a crucial role in our analysis. Splitting the problem into two subproblems determined by $e$ will be called a {\em branch} (for an example see  Fig.~\ref{fig:branching}). A newly formed subproblem may induce several further changes to $G$ and $F$ by Step~1, as shown in Fig.~\ref{fig:branching}. Here, and in the rest of this paper, the bold black line indicates an edge in $F$, the bold gray line an edge selected to be added to $F$ by the current branch or forced to be be added to $F$ in Step~1 directly after the branch, and the dotted line indicates an edge removed by the branch or forced not to appear in a cycle. 

To reduce the size of the recursion tree, the following prioritisations are used. First, if $G\setminus F$ contains a 4-cycle, two vertices of which are adjacent to edges in $F$, like here 
\begin{tikzpicture}[node distance=0.1, scale=0.3]
      \node[smallvertex] (a) at (0,0) {};
      \node[smallvertex] (b) [right=of a] {};
      \node[smallvertex] (c) [below=of b] {};
      \node[smallvertex] (d) [left=of c] {};
      \node[smallvertex] (e) [above left=of a] {};
      \node[smallvertex] (f) [above right=of b] {};
      \node[smallvertex] (g) [below right=of c] {};
      \node[smallvertex] (h) [below left=of d] {};
      \foreach \source/ \dest in {a/e, b/f}
           \path[smalledgef] (\source) -- (\dest);
     \foreach \source/ \dest in {a/b, b/c, c/d, d/a, d/h, c/g}
           \path[smalledge] (\source) -- (\dest);
    \end{tikzpicture}\ , 
then a non-cycle edge of $G\setminus F$ that is incident to a vertex of the cycle is chosen (Step~$3(a)$). If there is no such 4-cycle, but $F$ is nonempty, then an edge in $G\setminus F$ which is adjacent to an edge in $F$ is chosen (Step~3(b)). If $F$ is empty, then any edge in $G$ is chosen. 

% {$ $} 
% 
% \vspace*{-4.5mm}
%\noindent
To analyse the time complexity of the algorithm, Eppstein derives a recurrence
\begin{equation}\label{eq:eppsetin:recc}
 T(s) = \max \{2T(s-3), T(s-2)+T(s-5) \},
\end{equation}
the solution of which can be used to bound the number of iterations of the algorithm\footnote{In fact, Eppstein uses a slightly different recurrence which, however, has the same asymptotic solution. We use the form  above to be consistent with our analysis.}. The recurrence uses the variable $s$ defined as  
\begin{equation} \label{eq:eppstein:s}
  s=|V(G)|-|F|-|\mathcal{C}|,
\end{equation} 
where  $\mathcal{C}$ denotes the set of 4-cycles of $G$ that form connected components of $G\setminus F$. The solution is $T(s)=\caO(2^{s/3})$ and since in an $n$-vertex graph $s$ is at most $n$ this gives the bound $\caO(2^{n/3})$ on the run-time of the algorithm.

\section{Main Ideas of our Analysis}
\label{mainidea}

In this section we present basic notions and main ideas of our analysis of Eppstein's algorithm.
Our technique largely relies on exploiting ideas due to Iwama and Nakashima \cite{IwamaN07} and on a new, more careful analysis of the branching tree.

The worst-case component $2T(s-3)$ in the recurrence relation~(\ref{eq:eppsetin:recc}) corresponds to the situations when the algorithm reduces the problem of size $s$ to two subproblems each of size $s-3$. Eppstein proves that there are two cases leading to such situations. In the first case the algorithm chooses in Step~$3(b)$ an edge $yz$ to be adjacent to a forced edge $xy$ and neither $yz$ nor $yw$ -- the third edge of $G$ incident to $y$ -- is adjacent to a second edge in $F$. In the second case, an edge $yz$ chosen in Step~$3(b)$ belongs to a length-six cycle of unforced edges, each vertex of which is also incident to a forced edge. In their paper Iwama and Nakashima call the branches of the algorithm corresponding to these cases respectively $A$- and   $B$-branches. Branches which are neither of type $A$ nor $B$ are called $D$-branches. Results of performing these branches are shown in Fig.~\ref{fig:basic:branchings}.

%\vspace*{-4mm}
\begin{figure}[htbp]
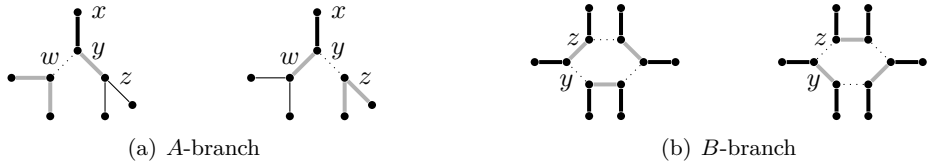

  \centering
  \subfigure[$A$-branch]{
    \begin{tikzpicture}[node distance=0.3*1.3, scale=1.5]
      \input{branch}

      \foreach \source/ \dest in {c2/z2, d2/z2, b3/x3, e3/x3}
        \path[edge] (\source) -- (\dest);
      \foreach \source/ \dest in {b2/x2, e2/x2, y2/z2, x3/y3, c3/z3, d3/z3}
        \path[edgesel] (\source) -- (\dest);
      \foreach \source/ \dest in {x2/y2, y3/z3}
        \path[edgerem] (\source) -- (\dest);
    \end{tikzpicture}
  }
  \hspace{1.5cm}
  \subfigure[$B$-branch]{
    \begin{tikzpicture}[node distance=0.3*1, scale=1.3]
      \input{6cycle-br}

      \foreach \source/ \dest in {a2/b2, c2/d2, e2/f2, b3/c3, d3/e3, f3/a3}
        \path[edgerem] (\source) -- (\dest);
      \foreach \source/ \dest in {b2/c2, d2/e2, f2/a2, a3/b3, c3/d3, e3/f3}
        \path[edgesel] (\source) -- (\dest);
      \node[label=left:{$z$}] at ($(a2)+(+0.1,0)$) {};
      \node[label=below:{$y$}] at ($(f2)+(0,+.06)$) {};
      \node[label=left:{$z$}] at ($(a3)+(+0.1,0)$) {};
      \node[label=below:{$y$}] at ($(f3)+(0,+.06)$) {};
      
    \end{tikzpicture}
  }
  \caption{Results of performing $A$- and $B$-branches.}
  \label{fig:basic:branchings}
\end{figure}

%\vspace*{-4mm}
Since $A$- and $B$-branches have the biggest impact on worst-case performance, the idea is to find an upper bound on the number of such branches. Iwama and Nakashima observed that performing an $A$-branch results in subproblems for which the number of free vertices in the graph is decreased by four. A vertex is called free if it is not incident to an edge in $F$. This implies that the total number of $A$-branches along any path of the backtrack tree cannot exceed $n/4$ for any $n$-vertex input graph. Thus, an important challenge now is to bound the worst-case number of occurrences of $B$-branches and then to incorporate these information by introducing new variables into a recurrence equation. The solution of the recurrence equation can then be used to bound the run-time of the algorithm.

In our approach we will count, similarly as in \cite{IwamaN07}, $A$-branches and $B$-branches together. We prove that if $P$ is a single path of the backtrack tree and $a$, resp. $b$, denotes the number of $A$-, resp. $B$-branches, then $3a+7b\le n$  holds for any $n$-vertex input graph. This bound will play a crucial role in our analysis. 

To prove the estimation above we use some ideas due to Iwama and Nakashima for analysing 6-cycles and  
a new, more careful technique to analyse $B$-branches. We show that along any path $P$ containing $b$ $B$-branches the algorithm selects at least $4b$ edges which are neither selected by $A$- nor by $B$-branches. Since any $A$- and any $B$-branch selects three edges, the inequality $3a+3b+4b\le n$ follows.

\section{Modification of Eppstein's Algorithm and the Analysis}
\label{section:our:analysis}

We make a minor modification in the algorithm of Eppstein, and add a prioritisation (introducing a new Step~$3(a')$ between Step~$3(a)$ and $3(b)$)
%(indicated in italic font below) 
destroying  structures leading to $B$-branches. A similar modification was used before by Iwama and Nakashima in \cite{IwamaN07}. The modification facilitates the analysis of the branching trees and enables to increase the lower bound on the number of edges selected by $D$-branches. 

A 6-cycle in $G$ is called {\em live} if none of its six cycle edges are selected, and a 6-cycle which is not live is called {\em dead}.

\begin{itemize}
\item[] \hspace*{-6mm}$3(a')$ If there is no 4-cycle and if $G\setminus F$ contains a live 6-cycle with a vertex $y$ which has a neighboring edge in $F$ (that is not a cycle edge but an attached one), let $z$ be one of $y$'s neighboring vertices (on the cycle). If two or more such live 6-cycles exist, then select a 6-cycle such that most attached edges are already selected. If there is more than one such edge $yz$ in the 6-cycle, choose $yz$ so, that $z$ also has a neighboring edge in $F$.
\end{itemize}

It is easy to see, that the correctness as well as the analysis presented by Eppstein are not affected by this modification. Now our aim is to provide a recurrence equation which incorporates information about $A$- and $B$-branches. The solution of this recurrence equation will then bound the run-time of the algorithm. %Next we will solve the recurrence. 

\subsection{The Recurrence} 
\label{subsection:Recurrence}
%\noindent{\bf The Recurrence.}
% The recurrence of the algorithm can be analyzed based on Step~3, which chooses an edge for the next recursion step. In one recursive call of the algorithm the edge is selected and in the other call the edge is removed from $G$. The branching into these two subproblems can be characterized by a set of different branching cases, which will be defined in the following.
For our analysis of the run-time of the algorithm, we define a multivariate recurrence equation  in the variables $n,s,x,y,$ and $f$. Variable $n$ denotes the number of vertices of the input graph and will not be modified by the recurrence. Variable $s$ is defined in a similar way as in \cite{Eppstein07} (cf. equation (\ref{eq:eppstein:s}) in a previous section); we let 
$$s=|V(G)|-|F|-2|\mathcal{C}|,$$ where  $\mathcal{C}$ is the set of 4-cycles  as defined in Section~\ref{secalgorithm}. 
%denotes the set of 4-cycles of $G$ that form connected components of $G\setminus F$. 
Next, let 
$$
  x=n/4-a \quad \text{and} \quad y=n/7 -b,
$$ 
where $a$, resp. $b$, is the number of $A$-branches, resp. $B$-branches, made by the algorithm along the current backtrack path\footnote{Speaking more formally, the algorithm can e.\,g.\ use two variables (starting with null values) to count the number of $A$- and $B$-branches and $a$ and $b$ are the current values of those variables.}. Finally, $f$ is the number of free vertices in $G$. We can bound the worst-case number of leaves of the backtrack tree as the solution of the following recurrence defined for non-negative integers $n,s,x,y,f$ as follows:
\begin{equation}
T(n,s,x,y,f)=\text{max}\left.\begin{cases}
              2T(n,s-3,x-1,y,f-4)\\
              2T(n,s-3,x,y-1,f)\\
              T(n,s-5,x,y,f-2)+T(n,s-2,x,y,f-2)\\
              2T(n,s-4,x,y,f)\\
              T(n,s-4,x,y,f-2)+T(n,s-3,x,y,f-2).
             \end{cases}\right.
\label{new:recurrence}
\end{equation}

The recurrence is based on the branch cases given in \cite{Eppstein07} and can be shown to  describe correctly the reductions of the problem. This can be done in an analogous way as in the proof of Lemma~7 in \cite{Eppstein07} with the modification that we now consider additionally free vertices (parameter $f$) in each situation. For Eppstein's case of $2T(s-3)$, we distinguish three subcases in the recurrence above: $A$- and $B$-branches are covered by the first, resp. second line, and branches caused by Step~$3(a)$ are included in the fifth line. Note, that due to different definitions of parameter $s$ in our recurrence the branch caused by Step~$3(a)$ decreases the variable $s$ in the reduced subproblems by 3 and by 4 (and not by 3 and 3 as in \cite{Eppstein07}). For Eppstein's case of $T(s-2)+T(s-5)$, we include the third and fourth line. He omits the case $T(s-4)$, since it is dominated by the other cases. 
% Formaly: -------------------

% Therefore, we can conclude the following about the recurrence~(\ref{new:recurrence}). Suppose that for all integers $n,s,x,f$ which specify the base cases of~(\ref{new:recurrence}) it is true that $T(n,s,x,f)$ bounds from above the number of leaves of the backtrack tree for any $n$-vertex graph when the algorithm starts with values $s,x,f$. Then this property holds for all  non-negative integers  $n,s,x,f$. 

To complete the definition of the recurrence we need to determine base cases describing termination conditions. Obviously, if $s=0$ then the algorithm reaches the bottom of the recurrence tree. Thus, we define the first base case as $T(n,0,x,y,f)=1$, for any $n,x,y,f$. The bottom of the recursion is also reached by the algorithm, if $s>0$ and no further branch can be applied. This implies the second base case: $T(n,s,x,y,f)=1$ if 
$\max\{ T(n,s-3,x-1,y,f-4),
        T(n,s-3,x,y-1,f),
        \min\{T(n,s-5,x,y,f-2), T(n,s-2,x,y,f-2)\},
        T(n,s-4,x,y,f),
        \min\{T(n,s-4,x,y,f-2), T(n,s-3,x,y,f-2)\}
     \}=0$.
A value of zero terminates impossible branches: one of these termination conditions is reached when $f<0$ or $s<0$. Thus, we define $T(n,s,x,y,f)=0$ if $f<0$ or $s<0$.

The last termination condition, which plays a crucial role in keeping the size of the tree small, occurs when $3x+7y<\frac{3}{4}n$. We define:
\begin{equation}
     T(n,s,x,y,f)=0 \quad \text{if} \quad 3x+7y<\frac{3}{4}n.
\label{new:recurrence:base:case}
\end{equation}

The correctness follows from our main technical result bounding the number of $A$- and $B$-branches along any path of the backtrack tree.

\begin{proposition} Let $P$ be a single path of the backtrack tree  and suppose that there are 
%in total 
$b$ $B$-branches on $P$. Then along $P$ the algorithm selects at least $4b$ edges which are neither selected by $A$- nor by $B$-branches. 
\label{dbranches}
\end{proposition}

We will prove this proposition in the next section. 

\begin{corollary}\label{cor:on:xyn}
The invariant of the algorithm running on an input graph with $n$ vertices is that $3a+7b\le n$ or equivalently that $3x+7y\ge \frac{3}{4}n$. Thus, (\ref{new:recurrence:base:case}) is the correct base case of recurrence~(\ref{new:recurrence}).
\end{corollary}
\begin{proof}
By definition, each $A$-branch and each $B$-branch selects 3 edges. Thus, for any single path $P$ of the backtrack tree from this property and from Proposition~\ref{dbranches} it follows that $3a+7b\le n$, where $a$, resp. $b$, denotes the number of $A$-branches, resp. $B$-branches, on $P$ and $n$ denotes the number of vertices of the input graph. The inequality above can be rewritten as $3x+7y\ge \frac{3}{4}n$ using the definitions $x=n/4-a$ and $y=n/7-b$. Thus $3x+7y < \frac{3}{4}n$ cannot occur. 
\end{proof}

Since, for any single path of the backtrack tree, the total number of $A$-branches cannot exceed $n/4$ and the total number of $B$-branches is not bigger then $n/7$, we can conclude:  
\begin{corollary}
The worst-case number of leaves of the backtrack tree on an $n$-vertex graph is bounded by $T(n,n,n/4,n/7,n)$.
\end{corollary}

\subsection{Solving the Recurrence}
%\noindent{\bf Solving the Recurrence.}
To bound the solution for the recurrence \eqref{new:recurrence} we use a function of the form
%\begin{equation}
$R(n,s,x,y,f)=2^{\alpha s+\beta (x+\frac{7}{3}y-\frac{1}{4}n)+\gamma f}.$
Note that the term $x+\frac{7}{3}y-\frac{1}{4}n$ incorporates the information provided by Proposition~\ref{dbranches} and Corollary~\ref{cor:on:xyn}.
%\label{recurrencesolution:new}
%\end{equation}
Our aim is to find parameters $\alpha, \beta,$ and $\gamma$ with the property that for all $n,s,x,y,f$ it holds: $T(n,s,x,y,f)\leq R(n,s,x,y,f)$, and that 
the following upper bound is best possible 
\[
 T(n,n,n/4,n/7,n) \le 2^{\alpha n+\beta (\frac{n}{4}+\frac{7}{3}\frac{n}{7}-\frac{1}{4}n)+\gamma n} =2^{(\alpha+\frac{\beta}{3}+\gamma)n}.
\]
Thus, we have to minimize
%\begin{equation}
$\alpha+\beta/3+\gamma$
%\label{miniwama}
%\end{equation}
with constraints  $T(n,s,x,y,f)\leq R(n,s,x,y,f)$. For the function $R$ of the form
as above %~(\ref{recurrencesolution:new}) 
we get first the following dependencies deduced from the recurrence equation \eqref{new:recurrence}:
\begin{equation*}
\begin{split}
2R(n,s-3,x-1,y,f-4)%&=2\cdot 2^{\alpha (s-3)+\beta((a-1)+2b)+\gamma(f-4)}\\
&=2^{(-3\alpha -\beta-4\gamma)+1}\, R(n,s,x,y,f)\\
2R(n,s-3,x,y-1,f)%&=2\cdot 2^{\alpha (s-3)+\beta(a+2(b-1))+\gamma f}\\
&=2^{(-3\alpha -\frac{7}{3}\beta)+1}\, R(n,s,x,y,f)\\
R(n,s-5,x,y,f-2)+R(n,s-2,x,y,f-2)%&=2^{\alpha (s-5)+\beta(a+2b)+\gamma(f-2)}+2^{\alpha (s-5)+\beta(a+2b)+\gamma(f-2)}\\
&=(2^{-5\alpha -2\gamma}+2^{-2\alpha -2\gamma})\, R(n,s,x,y,f)\\
2R(n,s-4,x,y,f)%&=2\cdot 2^{\alpha (s-4)+\beta(a+2b)+\gamma f}\\
&=2^{(-4\alpha )+1}\, R(n,s,x,y,f)\\
R(n,s-4,x,y,f-2)+R(n,s-3,x,y,f-2)%&=2^{\alpha (s-4)+\beta(a+2b)+\gamma(f-2)}+2^{\alpha (s-3)+\beta(a+2b)+\gamma(f-2)}\\
&=(2^{-4\alpha -2\gamma}+2^{-3\alpha -2\gamma})\, R(n,s,x,y,f).
\end{split}
\label{recurrenceequations}
\end{equation*}
From %\eqref{recurrenceequations}
this  we conclude that a solution is valid under the constraints
\begin{equation}
\begin{split}
&3\alpha+\beta+4\gamma\geq 1,\quad 3\alpha+\frac{7}{3}\beta\geq 1,\quad 2^{-5\alpha-2\gamma}+2^{-2\alpha-2\gamma}\leq 1,\\
&4\alpha\geq 1,\quad 2^{-4\alpha-2\gamma}+2^{-3\alpha-2\gamma}\leq 1.
\end{split}
\label{constraints}
\end{equation}
Minimizing $\alpha+\beta/3+\gamma$ under \eqref{constraints} gives a rational approximation
%\begin{equation}
$\alpha=\frac{157}{531},\ \beta=\frac{1-3\alpha}{2}=\frac{20}{413},$ and $\gamma=\frac{\beta}{3}=\frac{20}{1239}$
%\end{equation}
resulting in 
\begin{equation*}
\alpha+\frac{\beta}{3}+\gamma=\frac{1219}{3717}\quad \text{and}\quad2^\frac{1219}{3717}\approx 1.25523.
\label{newupperbound}
\end{equation*}

From this estimation we can conclude an upper bound  $\mathcal{O}(1.2553^n)$ on the run-time of the algorithm.
%concluded to an upper bound for our new analysis of $\mathcal{O}(1.2553^n)$, thus improving the bound of $\mathcal{O}(1.260^n)$ as found in \cite{Eppstein07}.
Since the modified algorithm can still use only linear space, this completes the proof of Theorem~\ref{theorem:main}.

\section{Proof of Proposition~\ref{dbranches}}
\label{sec:proofs}

This section is organized as follows. We start with some preliminaries. Next, Subsection~\ref{subsection:main:steps} presents main steps of the proof and 
in Subsection~\ref{subsection:proofs}  proofs of our key lemmata are given. Finally,  Subsection~\ref{subsection:auxiliary} provides an auxiliary result needed for the proofs.

\subsection{Preliminary Observations}
\label{subsection:preliminaries}
Recall, that any $B$-branch is caused by a length-six cycle of unforced edges, each vertex of which is also incident to a forced edge (i.\,e.\ a selected edge stored in $F$). Let $C(i)$, with $0 \le i \le 6$, denote a live~6-cycle such that $i$ edges of its six attached edges have already been selected. Thus, each $B$-branch is caused by a $C(6)$-cycle. We start with the following fact concerning $C(i)$-cycles which was observed by Iwama and Nakashima in the proof of Lemma~1 in \cite{IwamaN07}.
%Iwama and Nakashima 
%(see the proof of Lemma~1 in \cite{IwamaN07}).
\begin{lemma}[\cite{IwamaN07}] \label{ref:lemma:d:branch}
Let $C$ be a  (live) $C(3)$-, $C(4)$- or $C(5)$-cycle and assume a single branch of the algorithm increases the number of selected edges that are incident to $C$ but still leaves $C$ live (thus, e.\,g.\ $C(4)$ changes to $C(5)$ or $C(6)$). Then the branch is a $D$-branch. 
\label{onlyd}
\label{c3andmore}
\end{lemma}
\begin{proof} %[Proof of Lemma~\ref{c3andmore}]
Let $C$ be a $C(i)$-cycle as stated in the lemma and assume $Q$ is a branch increasing the number of selected edges that are incident to $C$ and still leaves $C$ live. If $Q$ is performed in Step~$3(a)$ then we are done since branches performed due to this step are $D$-branches. If Step~$3(a)$ cannot be applied, $Q$ has to be done due to Step~$3(a')$ since in the current situation there exists at least one live 6-cycle, namely $C$. But to leave $C$ live, there has to exist at least one additional (live) $C(j)$-cycle, with $j\ge i\ge 3$, since otherwise $Q$ would transform $C$ making it dead. 

Thus, assume $Q$ affects in Step~$3(a')$ a $C(j)$-cycle transforms cycle $C$ from $C(i)$ to $C(i')$, with $i'>i$. Obviously, $Q$ cannot be a $B$-branch because such branches cannot change the degree of another $C(i)$-cycle. We will show that $Q$ cannot be an $A$-branch as well, which will complete the proof of the lemma.

\begin{figure}[htbp]
  \centering
  \begin{tikzpicture}[node distance=0.5*1, scale=1]
    \node[vertex,label=left:$z$] (a1) at (0,0) {};
    \node[vertex,label=left:$z$] (a2) at (5,0) {};

    \foreach \ind in {1,2} {
      \input{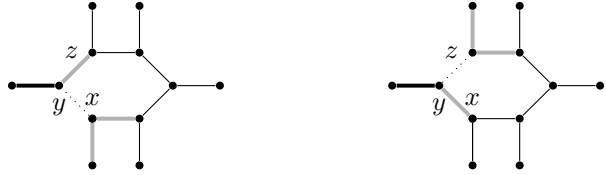}
%      \node[vertex,label=above left:$y$] (f\ind) [above left=of e\ind] {};
      \node[label=below:{$y$}] at ($(f\ind)+(0,+.07)$) {};
      \node[vertex,label=above:$x$] (e\ind) [left=of d\ind] {};
    }

    \foreach \source/ \dest in {a1/g1, b1/h1, c1/i1, d1/j1, e1/k1, a1/b1, b1/c1, c1/d1, b2/h2, c2/i2, d2/j2, e2/k2, b2/c2, c2/d2, d2/e2, e2/k2}
      \path[edge] (\source) -- (\dest);
    \foreach \source/ \dest in {f1/l1, f2/l2}
      \path[edgef] (\source) -- (\dest);
    \foreach \source/ \dest in {e1/f1, f2/a2}
      \path[edgerem] (\source) -- (\dest);
    \foreach \source/ \dest in {f1/a1, d1/e1, e1/k1, a2/g2, a2/b2, e2/f2}
      \path[edgesel] (\source) -- (\dest);
  \end{tikzpicture}
  \caption{Results of performing an $A$-branch affecting  a live 6-cycle $C'$: they show that at least four of the  vericies of $C'$ have to be free. Thus, right before the branch is performed, $C'$  has to be a $C(j)$-cycle, with $1\le j\le 2$.}
  \label{6cycle-a-branch}
\end{figure}

Assume, to the contrary, that $Q$ is an $A$-branch. Figure~\ref{6cycle-a-branch} shows results of performing an $A$-branch affecting a live 6-cycle $C'$. We can see, that at least four of the vertices of $C'$ have to be free, since otherwise the branch would select more than three edges. Thus, it follows directly that $C'$ can be at most a $C(2)$-cycle. We get a contradiction, since branch $Q$ has to transform a $C(j)$-cycle, with $j\ge 3$. 
\end{proof}

The simple fact below has also been used in \cite{IwamaN07}, although it was not stated explicitly in the paper.
\begin{wrapfigure}{r}{0pt}
%\begin{figure}[htbp]
  \centering
   \hspace*{2mm}
   \begin{tikzpicture}[node distance=0.25]
      \tikzstyle{edge4} = [draw,color=black!30,line width=1.5pt,-,color=red]
  
      \node[vertex] (a1) {};
      \foreach \ind in {1}
        \input{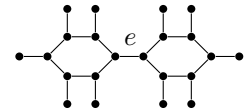};

      \foreach \source/ \dest in {}
	\path[edgef] (\source) -- (\dest);
      \foreach \source/ \dest in {f1/l1, a1/g1, b1/h1, a1/b1, b1/c1, c1/d1, d1/e1, e1/f1, f1/a1, c1/i1, d1/j1, e1/k1}
	\path[edge] (\source) -- (\dest);

      \node[vertex] (a2) [above right=of i1] {};
      \node[vertex] (b2) [right=of a2] {};
      \node[vertex] (c2) [below right=of b2] {};
      \node[vertex] (d2) [below left=of c2] {};
      \node[vertex] (e2) [left=of d2] {};

      \node[vertex] (g2) [above=of a2] {};
      \node[vertex] (h2) [above=of b2] {};
      \node[vertex] (i2) [right=of c2] {};
      \node[vertex] (j2) [below=of d2] {};
      \node[vertex] (k2) [below=of e2] {};
 
     \foreach \source/ \dest in {}
	\path[edgef] (\source) -- (\dest);
      \foreach \source/ \dest in {d2/j2, e2/k2, a2/g2, c2/i2, a2/b2, b2/c2, c2/d2, d2/e2, i1/a2, b2/h2, e2/i1}
	\path[edge] (\source) -- (\dest);
                         
     \node[label=above:$e$] at ($(c1)+(0.2,-0.1)$) {};
    \end{tikzpicture}   
    \hspace*{2mm}
  \caption{A common attached  edge.} 
  \label{fig:commen:edge} 
%\end{figure}
\end{wrapfigure}
\begin{fact} \label{fact:attached:edges}
Let $P$ be a single path of the backtrack tree and let $C$ and $C'$ be any cycles which become $C(6)$ on $P$. Then there is no edge in $G\setminus F$ which is attached both to $C$ and to $C'$ and which is added to $F$ by the algorithm. Moreover, the cycles $C$ and $C'$ are disjoint, i.\,e.\ they have no common cycle edges.
\end{fact}

\begin{proof}
Assume $e$ is an attached edge to both $C$ and $C'$ when it is added to $F$, i.\,e.\ that the  situation is like the one shown in Fig.~\ref{fig:commen:edge}. Then it is easy to see that this can happen only in the case when the initial branch chooses the edge $e$ and next adds it to $F$. But after this step $C$ or $C'$ becomes dead before reaching $C(6)$. Thus, we get a contradiction that $C$ and $C'$ are $C(6)$ somewhere on $P$. 

It is also easy to see that cycles $C$ and $C'$ cannot have common cycle edges. Otherwise the first transition which makes one of such cycles a $C(6)$-cycle makes at the same time the other cycle dead before it becomes $C(6)$ -- a contradiction.
\end{proof}

Using the fact above one can show, for example, that for any single path $P$ if there are $b$ $B$-branches on $P$ then along $P$ the algorithm selects at least $6b+3b$ edges. Thus, one can conclude that for any $n$-vertex graph, along any path $P$ at most $n/9$ $B$-branches can occur. This is a better bound than $n/7$ shown in Section~\ref{subsection:Recurrence} but it seems to be useless to our analysis. From Fact~\ref{fact:attached:edges} we can conclude also:
\begin{corollary}\label{cor:b:branch}
Let $P$ be a single path of the backtrack tree and let $C$ be any cycle which became $C(6)$ on $P$. Then any attached edge of $C$ is selected by an $A$- or $D$-branch, but not by a $B$-branch.
\end{corollary}

\subsection{Outline of the Proof of Proposition~\ref{dbranches}}
\label{subsection:main:steps}

To prove the proposition we will associate with each $C(6)$-cycle on a path $P$ at least four edges which are selected neither by $A$- nor by $B$-branches. Fact~\ref{fact:attached:edges} and Corollary~\ref{cor:b:branch} suggest the following approach, which in fact was used in \cite{IwamaN07}: for any 6-cycle $C$ which becomes $C(6)$ on $P$, count the number of attached edges of $C$ which have been selected by $D$-branches. Using Lemma~\ref{ref:lemma:d:branch} one can show easily that for any such cycle there exists at least one edge selected by a $D$-branch. In their paper \cite{IwamaN07} Iwama and Nakashima claim that the number of edges can be increased to three. Unfortunately, this claim is false and one can show that there exist cubic graphs such that for some paths $P$ there exists a $C(6)$-cycle on $P$ having only two attached edges selected by $D$-branches and the remaining four edges are selected by $A$-branches (for more details see Section~\ref{section:comments:on:IN}). 

Thus, to associate with each $C(6)$-cycle at least four edges selected by $D$-branches, we need to extend   significantly the approach above. We will consider different types of cycles. Obviously, if a $C(6)$-cycle has already at least four edges selected by $D$-branches (we will call such cycles $B_1$-cycles), then we are done. 
The bad cases, i.\,e.\ cycles having more than two attached edges selected by $A$-branches, will be divided into two subcases called  $B_2$- and $B_3$-cycles.

\begin{definition}
Let $P$ be a single path of the backtrack tree and let $C$ be a $C(6)$-cycle somewhere on $P$. We call $C$ a {\em $B_1$-cycle} if at most two attached edges have been selected by $A$-branches. We call $C$ a {\em $B_2$-cycle} if more than two attached edges have been selected by $A$-branches and it changes on $P$ from $C(0)$ to $C(3)$ and then from $C(3)$ to $C(6)$ (directly or indirectly). We call $C$ a {\em $B_3$-cycle} if more than two edges have been selected by $A$-branches and $C$ becomes $C(6)$ by changes on $P$ from $C(0)$ to $C(6)$ without becoming $C(3)$ in between. We call an edge a {\em $B$-attached edge} if it is an edge attached to a $B_1$, $B_2$ or $B_3$-cycle.
\end{definition}

Below we illustrate how $B_2$- and $B_3$-cycles can transform from $C(0)$ to $C(6)$ (a letter $A$, resp. $D$, indicates a direct $A$-, resp. $D$-branch):
\[
\begin{array}{ll}
%\text{$B_1$-cycles}:&C(0) \overset{D}{\longrightarrow} C(2) \overset{A}{\longrightarrow} C(3) \overset{D}{\longrightarrow} C(6)\\
%                    &C(0) \overset{A}{\longrightarrow} C(2) \overset{D}{\longrightarrow} C(4) \overset{D}{\longrightarrow} C(5)\overset{D}{\longrightarrow} C(6)\\
\text{$B_2$-cycle}:&C(0) \overset{A}{\longrightarrow} C(1) \overset{A}{\longrightarrow} {\mathbf{C(3)}} \overset{D}{\longrightarrow} C(4)\overset{D}{\longrightarrow} C(6)\\
\text{$B_3$-cycle}:&C(0) \overset{A}{\longrightarrow} {\mathbf{C(2)}} \overset{A}{\longrightarrow}  {\mathbf{C(4)}} \overset{D}{\longrightarrow} C(6)
\end{array}
\]

Now, assume $P$ is a path of the backtrack tree which contains $b_i$ $B_i$-cycles, for $i=1,2,3$. Our aim is to prove that there exist  $4b_1+4b_2+4b_3$ edges on $P$ which are selected by $D$-branches. We summarize first the case of $B_1$-cycles. 

\begin{lemma}[$B_1$-cycles]
Let $P$ be a single path of the backtrack tree and let $C$ be a $B_1$-cycle on $P$. 
Then four attached edges of $C$ have been selected by $D$-branches.
\label{lemma4b1}
\end{lemma}

Since $B$-branches do not select an attached edge of another live 6-cycle, the lemma follows directly from the definition of $B_1$-cycles. Thus, the first intermediate conclusion is that if on a path $P$ there are $b_1$ $B_1$-cycles then there exist at least $4b_1$ edges selected by $D$-branches. This follows from the lemma above and Fact~\ref{fact:attached:edges} that we do not count the edges twice.

The tricky part of the proof involves $B_2$- and $B_3$-cycles.

\begin{definition}
Let $P$ be a single path of the backtrack tree. A $B_2$-cycle $C$ is called {\em active} somewhere on $P$, if $C$ is $C(3)$, $C(4)$ or $C(5)$. A $B_3$-cycle $C$ is called {\em active} somewhere on $P$, if $C$ is $C(4)$ or $C(5)$.
Additionally, we say that a cycle $C$ is {\em activated} by a branch $Q$, if $C$ was not active right before $Q$ and is active right after performing $Q$.
\end{definition}

Now, we proceed as follows. We follow path $P$ from the root to a leaf and analyse $B_2$- and $B_3$-cycles occurring along $P$. The needed edges are identified when a currently analysed $B_2$-, resp. $B_3$-cycle, is active. The lemma below guarantees that we do not count edges twice, i.\,e.\ that all of the edges we find are pairwise disjoint.

\begin{lemma}
Let $P$ be a single path of the backtrack tree. Then it is true that (1) at any time there are at most two $B_2$-cycles or at most one $B_3$-cycle active and (2) if some $B_2$- or $B_3$-cycle is active, no branch can make another cycle active.
% (for an example of an  arrangement of active $B_2$- (white) and $B_3$-cycles (gray) along a path~$P$, see fig. below)
(See Fig.~\ref{fig:b2b3structure}.).
\label{b2b3structure}
\end{lemma}
%\vspace*{-6mm}
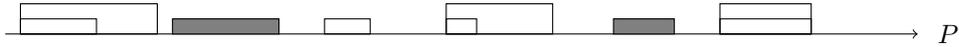
\begin{figure}[htbp]
  \centering
  \begin{tikzpicture}[scale=0.2]
    \draw[->] (0,0) -- (60,0);
    \node at (62,0) {$P$};
    \draw (1,0) rectangle (6,1);
    \draw (1,0) rectangle (10,2);
    \draw[fill=gray] (11,0) rectangle (18,1);
    \draw (21,0) rectangle (24,1);
    \draw (29,0) rectangle (31,1);
    \draw (29,0) rectangle (36,2);
    \draw[fill=gray] (40,0) rectangle (44,1);
    \draw (47,0) rectangle (53,1);
    \draw (47,0) rectangle (53,2);
  \end{tikzpicture}
  
  \caption{An arrangement of active $B_2$-cycles  (white) and $B_3$-cycles (gray) along a path~$P$.} 
  \label{fig:b2b3structure}
\end{figure}

Note that from the lemma above follows that $B_2$-cycles and $B_3$-cycles cannot be active simultaneously on $P$.
Before we show the main result of this section, we define an {\em internal edge} as a selected edge with certain properties.
\begin{definition}\label{def:internal:edge}
Let $P$ be a single path and $Q_1,Q_2,\ldots$ denote the branches along $P$ which are not $B$-branches. Let $E_j$ be the set of edges selected as direct consequence of branch $Q_j$ along $P$. These are the edges selected by the branch itself and edges selected by subsequent iterations of Step~1. Note that the contraction of adjacent edges does not modify $E_j$. Be further $E_P=\bigcup_j E_j$. Then an edge is called an \emph{internal edge}, if it is $(i)$ selected by a $D$-branch and has two adjacent edges in $E_P$ or $(ii)$ selected  by Step~2 to a resulting Hamiltonian cycle.
\end{definition}

\begin{figure}[htbp]
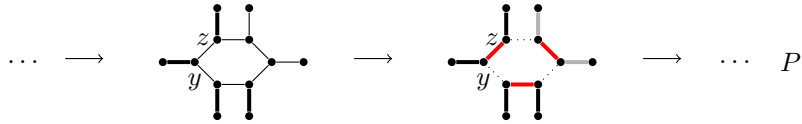

  \centering
   \begin{tikzpicture}[node distance=0.3]
      \tikzstyle{edge4} = [draw,color=black!30,line width=1.5pt,-,color=red]
  
      \node[vertex] (a1) {};
      \foreach \ind in {1}
        \input{6cycle-comp};

      \node[label=left:{$z$}] at ($(a1)+(+0.15,0)$) {};
      \node[label=below:{$y$}] at ($(f1)+(0,+.1)$) {};
      \foreach \source/ \dest in {d1/j1, e1/k1, f1/l1, a1/g1}
	\path[edgef] (\source) -- (\dest);
      \foreach \source/ \dest in {a1/b1, b1/c1, c1/d1, d1/e1, e1/f1, f1/a1, b1/h1, c1/i1}
	\path[edge] (\source) -- (\dest);

      \node[vertex] (a2) at (3.8,0) {};
      \foreach \ind in {2}
        \input{6cycle-comp};
      \node[label=left:{$z$}] at ($(a2)+(+0.15,0)$) {};
      \node[label=below:{$y$}] at ($(f2)+(0,+.1)$) {};
      \foreach \source/ \dest in {d2/j2, e2/k2, f2/l2, a2/g2}
	\path[edgef] (\source) -- (\dest);
      \foreach \source/ \dest in {a2/b2, c2/d2, e2/f2}
        \path[edgerem] (\source) -- (\dest);
      \foreach \source/ \dest in {c2/i2, b2/h2}
        \path[edgesel] (\source) -- (\dest);
      \foreach \source/ \dest in {b2/c2, d2/e2, f2/a2}
        \path[edge4] (\source) -- (\dest);

      \draw [->] (-2,-0.25) -- (-1.5,-0.25);
      \draw [->] (1.8,-0.25) -- (2.3,-0.25);
      \draw [->] (5.6,-0.25) -- (6.1,-0.25);
     % \draw [->] (2,-1.3) -- (2.5,-1.6);

     %\node[label=left:{$\ldots$}] at ($(l1)+(-1.5,0)$) {};

     \node[label=left:{\hspace*{-35mm} $\ldots$}] at (l1) {};
     \node[label=right:{\hspace*{13mm} $\ldots$}] at (i2) {};
     \node[label=right:{\hspace*{21mm} $P$}] at (i2) {};
    \end{tikzpicture}
  \caption{An example for internal edges (red) selected by a $D$-branch due to Step~$3(a')$. } 
  \label{fig:internal:edge} 
\end{figure}

For an example of internal edges see Fig.~\ref{fig:internal:edge}.  Along a path $P$ the algorithm chooses an edge $yz$ in a $C(4)$-cycle (left) in Step~$3(a')$. Note that this leads to a $D$-branch. The figure shows a fragment of a path following a recursive call on $G, F\cup\{yz\}$. The configuration to the right shows the situation after the branch (including subsequent iterations of Step~1). The red edges are internal since they are selected by a $D$-branch and each red edge has two adjacent edges which could not be selected by a $B$-branch. 

An essential property of an internal edge is, that it is neither a $B$-attached edge nor an edge selected by $A$- or $B$-branches.

\begin{lemma}[$B_2$-cycles]
Let $P$ be a single path of the backtrack tree and suppose there are $b_2$ $B_2$-cycles on $P$. Then there exist at least $b_2$ internal edges for $P$.
\label{lemma4b2}
\end{lemma}

By definition, we know that any $B_2$-cycle becomes $C(3)$ before becoming $C(6)$. Hence, from Lemma~\ref{c3andmore} we get that any $B_2$-cycle has three attached edges selected by $D$-branches. Since an internal edge can be neither a $B$-attached edge nor an edge selected by $A$- or $B$-branches, from   Lemma~\ref{lemma4b2}  and from Lemma~\ref{lemma4b1} we can conclude that if on a path $P$ there are $b_1$ $B_1$-cycles and $b_2$ $B_2$-cycles, then there exist at least $4b_1+4b_2$ edges neither selected by  $A$- nor $B$-branches. To complete the proof we have to consider the case of $B_3$-cycles.

\begin{lemma}[$B_3$-cycles]
Let $P$ be a single path of the backtrack tree and suppose that we have a $B_3$-cycle, say $C$, somewhere on $P$. Then it holds that (1) two attached edges of $C$ have been selected by $D$-branches and (2) there exist two additional edges that are either internal edges or they are selected by Step~2. Moreover, the additional edges are pairwise disjoint for all such $C(6)$-cycles.
\label{lemma4b3}
\end{lemma}

Now we are ready to complete the proof. 
Let $P$ be a single path of the backtrack tree and suppose there are $b_i$ $B_i$-cycles, for $i=1,2,3,$ on $P$. From Lemma~\ref{lemma4b1} and Lemma~\ref{lemma4b2} we already know that there exist at least $4b_1+4b_2$ edges neither selected by  $A$- nor $B$-branches. Moreover, for $B_3$-branches we can conclude by Lemma~\ref{lemma4b3} and  by Lemma~\ref{b2b3structure} that there exist additionally at least $4b_3$ edges which are neither selected by $A$- nor by $B$-branches. Since, again by Lemma~\ref{b2b3structure} and Fact~\ref{fact:attached:edges}, the set of those $4b_3$ edges is disjoint with the edges assigned for $B_1$- and $B_2$-cycles, Proposition~\ref{dbranches} follows.\qed

% To prove Proposition~\ref{dbranches} we have to conclude that along $P$ there are $4b_1+4b_2+4b_3$ edges which are neither selected by $A$- nor by $B$-branches. Indeed, by Lemma~\ref{lemma4b1} we get that for $B_1$-branches there are $4b_1$ such edges. 
% 
% For $B_2$ we know from Lemma~\ref{c3andmore} that there are at least $3b_2$ edges selected by $D$-branches attached to $C(6)$-cycles causing the $B_2$-branches and from Lemma~\ref{lemma4b2} we have that there exist $b_2$ internal edges which are neither selected by $A$- nor by $B$-branches. Finally, for $B_3$-branches we can conclude by Lemma~\ref{lemma4b3} that there exist at lest $4b_3$ edges which are neither selected by $A$- nor by $B$-branches. Since, by Lemma~\ref{b2b3structure} all those edges  are pairwise disjoint, this completes the proof of Proposition~\ref{dbranches}.

\subsection{Proofs of Lemmata~\ref{b2b3structure}, \ref{lemma4b2} and \ref{lemma4b3}} 
\label{subsection:proofs}

\begin{proof}[Proof of Lemma~\ref{b2b3structure}]
We first claim that cycles can be activated only by $A$-branches. By definition a $B_2$-cycle $C$ changes on $P$ from $C(0)$ to $C(3)$ and then from $C(3)$ to $C(6)$ and it is activated, when it becomes $C(3)$. Due to Lemma~\ref{onlyd} the last three attached edges are selected by $D$-branches. Since a $B_2$-cycle has more than two attached edges selected by $A$-branches, the first three attached edges have to be selected by $A$-branches when the cycle was transformed form $C(0)$ to $C(3)$ and the claim follows. In case of $B_3$-cycles that, recall never become $C(3)$, we can conclude that such a cycle $C$ gets activated by a direct transition  $C(i)\rightarrow C(j)$, with $0\le i\le 2$ and $4\le j\le 6$. Since $C$ has more than two attached edges selected by $A$-branches and since, from Lemma~\ref{onlyd}, the last $6-j$ attached edges have to be selected by $D$ branches, we get that the transition $C(i)\rightarrow C(j)$ has to be performed due to an $A$-branch $Q$.

The second part {\it (2)} of the lemma follows now easily. 
%If $Q$ is an $A$-branch, there is no $C(i)$-cycle, with $i\geq 3$, right before $Q$ is performed, which remains live after this branch. This follows from two facts. 
Firstly, as shown in the proof of Lemma~\ref{ref:lemma:d:branch} (cf. Fig.~\ref{6cycle-a-branch}), if an $A$-branch affects a $C(j)$-cycle then $j\le 2$. Secondly, due to the prioritisation used in Step~$3(a')$, if the algorithm selects in this step a $C(j)$-cycle, with $j\le 2$, then there exists no $C(i)$-cycle, with $i>j$. This implies that, at the moment $Q$ is performed, there is no active $B_2$- or $B_3$-cycle. Consequently, no $A$-branch can be performed when a $B_2$- or $B_3$-cycle remains active and thus, no other cycle can be activated at that time. Therefore all $B_2$- and $B_3$-cycles which are active at the same time have to be activated by a single $A$-branch $Q$.

A proof of the first part {\it (1)} requires more involved arguments. We show first that by an $A$-branch only one $B_3$-cycle can be activated. For this aim we use the property that the only way in which the algorithm makes a $B_3$-cycle active is via an $A$-branch transition $C(2)\rightarrow C(4)$ which has to be directly followed by a branch due to Step~$3(a)$. This property is stated in Lemma~\ref{b3properties} and proven in the next subsection. 

Thus, activating a $B_3$-cycle $C$, an $A$-branch $Q$ attaches two of three selected edges to $C$. Since attached edges of $C(6)$-cycles are disjoint, these two edges do not activate another cycle. Moreover, from the property above we know that the third edge has to be attached to a 4-cycle, let us call it $\tilde{C}$, which forces to perform Step~$3(a)$ acting on $\tilde{C}$ (see Fig.~\ref{step3a} for configurations resulting by application of Step~$3(a)$ and the recursive calls). To see that the third edge cannot be attached to some  $B_2$- or $B_3$-cycle $C'$, too, we proceed as follows. 

Assume, for a contradiction, that the 4-cycle $\tilde{C}$ and a $B_2$- or $B_3$-cycle $C'$ share an attached edge $e$ selected by $Q$ and that, by the property above, the next branch performed by the algorithm is due to Step~$3(a)$ acting on $\tilde{C}$. We get that  $\tilde{C}$ and $C'$ also share two cycle edges adjacent to $e$. Further there is a vertex in which the cycles split up, meaning there is a cycle edge of $C'$ which is adjacent to  $\tilde{C}$. Since Step~$3(a)$ acting on $\tilde{C}$ is performed next, two cases can occur. Firstly, if the algorithm is recursively called on $G, F\cup\{yz\}$, where $yz$ is an edge chosen in Step~$3(a)$ (see Fig.~\ref{step3a:first:call}), all attached edges of $\tilde{C}$ are selected and thereby also cycle edge of $C'$. Hence, $C'$ becomes dead and cannot become $C(6)$ anymore.  Secondly, if the algorithm is recursively  called on $G\setminus\{yz\}, F$ (see Fig.~\ref{step3a:second:call}), the cycle edges of $\tilde{C}$ are either selected or removed, thus making  
$C'$ dead, too.  

Thus, we summarize that a single $A$-branch activates a $B_3$-cycle by selecting two attached edges,  $B_2$- and $B_3$-cycles do not share any attached edge, and finally that the third edge selected by the $A$-branch  cannot be an attached edge of any further $B_2$- or $B_3$-cycle.
Therefore we can conclude that by a single $A$-branch only one $B_3$-cycle can be activated.

It remains to show, that at most two $B_2$-cycles can be activated by an $A$-branch $Q$. Assume, for a contradiction, that $Q$ activates three $B_2$ cycles. 
From Fact~\ref{fact:attached:edges} we know that the cycles have to be pairwise disjoint, i.\,e.\ they cannot have common edges.
%It is easy to see that the cycles have to be pairwise disjoint, i.\,e.\ they cannot have common edges. If not, then the first transition which makes one of such cycles a $C(6)$-cycle makes at the same time the other cycle dead before it becomes  $C(6)$. 
%This contradicts that the second cycle is a  $B_2$-cycle.

\begin{figure}[htbp]
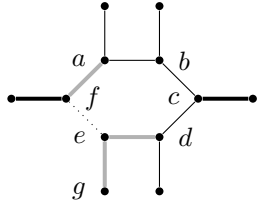

  \centering
    \begin{tikzpicture}[node distance=0.6*1, scale=1]
      \input{6cycle}

      \input{a-in-c2-labels}

      \foreach \source/ \dest in {a/g, b/h, d/j, e/k, a/b, b/c, c/d}
        \path[edge] (\source) -- (\dest);
      \foreach \source/ \dest in {f/l, c/i}
        \path[edgef] (\source) -- (\dest);
      \foreach \source/ \dest in {e/f}
        \path[edgerem] (\source) -- (\dest);
      \foreach \source/ \dest in {f/a, d/e, e/k}
        \path[edgesel] (\source) -- (\dest);
    \end{tikzpicture}
  \caption{An $A$-branch acting on  a $C(2)$-cycle.}
  \label{fig:thtee:b:two}
\end{figure}

Next, since $Q$ selects only three edges, all of the three $B_2$-cycles have to be $C(2)$ right before the branch and $Q$ acts on a fourth $C(2)$-cycle, which in turn will become dead after the branch (see Fig.~\ref{fig:thtee:b:two}).
%If one of the three $B_2$-cycles is not $C(2)$ then one can show easily the cycles are not disjoint.
% Now, the crucial property is, if an $A$-branch affects a $C(2)$-cycle, then the two selected edges have to be of distance three, as shown in Fig.~\ref{6cycle-a-branch-in-c2--6cycle-a-branch-for-c5}(a). 
However, for this branch case it is not possible to attach the three edges selected by $Q$ at three disjoint $C(2)$-cycles. In particular the edges $fa$ and $ed$ in Fig.~\ref{fig:thtee:b:two} cannot be attached to two disjoint $C(2)$-cycles. Therefore, we get a contradiction and conclude that we can activate at most two $B_2$-cycles by a single $A$-branch $Q$. Note, that by a such branch it can still be possible to get more than two non-disjoint $C(3)$-, $C(4)$- or $C(5)$-cycles.
\end{proof}

\begin{proof}[Proof of Lemma~\ref{lemma4b2}]
Let $P$ be a single path of the backtrack tree. 
Due to Lemma~\ref{b2b3structure} we know, that at most two $B_2$-cycles can be active at the same time and have to become dead, before other $B_2$-cycles can be activated. Let $Q$ be a fixed branch which activates one or two $B_2$-cycles. We consider the maximal subpath of the path $P$ which starts with $Q$ and along which at least one of the active $B_2$-cycles is still live. 
%which are performed until these $B_2$-cycles become dead.  
Along this subpath we examine all subsequent branches which transit an active $B_2$-cycle from $C(i)$ to $C(j)$, with $3\le i < j \le 6$. 
%are performed until these $B_2$-cycles become dead.  
We denote them by $Q_1,Q_2,\ldots,Q_t$. 
From Lemma~\ref{onlyd} we know that all these branches  are $D$-branches. An additional important property is that the branches $Q_\ell$ are performed by Step~$3(a)$ or $3(a')$, but not by Step~$3(b)$. 

Our aim is to show that during by $Q_1,Q_2,\ldots,Q_t$ at least two internal edges are selected if $Q$ activates two $B_2$-cycles and at least one internal edge is selected if $Q$ activates a single $B_2$-cycle. By the property above, to prove this claim it is sufficient  to analyze changes due to any possible branch performed by Step~$3(a)$ or $3(a')$. 
A resulting situation after completing such a branch, including subsequent iterations of Step~1, will be called a {\em configuration}, for short.

%the branches $Q_1,Q_2,\ldots,Q_t$, or if only one $B_2$-cycle is activated, at least one internal edge is selected. for two $B_2$-cycles activated by $Q$ at least two internal edge are selected by the branches $Q_1,Q_2,\ldots,Q_t$, or if only one $B_2$-cycle is activated, at least one internal edge is selected.

%Each branching forms two subproblems. Let the term {\em configuration} denote the changes made in one of the subproblems.
%-----------------[node distance=0.3*1, scale=1]
%\begin{floatingfigure}[r]{3cm}
\begin{figure}[htbp]
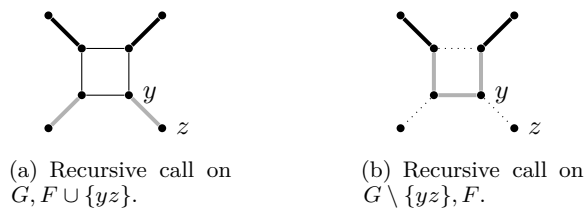

  \centering
  \subfigure[Recursive call on $G, F\cup\{yz\}$.]{
  \hspace{2mm}
  \begin{tikzpicture}[node distance=0.5, scale=1.4]
    \node[vertex] (a1) at (0,0) {};
    \foreach \ind in {1}
     \input{4cycle-comp}
     \foreach \source/ \dest in {a1/b1, b1/c1, c1/d1, d1/a1}
        \path[edge] (\source) -- (\dest);
     \foreach \source/ \dest in {a1/e1, b1/f1}
          \path[edgef] (\source) -- (\dest);
    \foreach \source/ \dest in {c1/g1, d1/h1}
       \path[edgesel] (\source) -- (\dest);
  \end{tikzpicture}
  \label{step3a:first:call}
  \hspace{2mm}
  }
  \hspace*{1.5cm}
  \subfigure[Recursive call on $G\setminus\{yz\}, F$.]{
  \hspace{2mm}
  \begin{tikzpicture}[node distance=0.5, scale=1.4]
    \node[vertex] (a1) at (0,0) {};
    \foreach \ind in {1}
     \input{4cycle-comp}
     \foreach \source/ \dest in {a1/e1, b1/f1}
          \path[edgef] (\source) -- (\dest);
    \foreach \source/ \dest in {b1/c1, c1/d1, d1/a1}
       \path[edgesel] (\source) -- (\dest);
   \foreach \source/ \dest in {a1/b1, c1/g1, d1/h1}
      \path[edgerem] (\source) -- (\dest);
  \end{tikzpicture}
  \label{step3a:second:call}
  \hspace{2mm}
  }
%   \begin{tikzpicture}[node distance=0.35, scale=1.4]
%     \input{4cycle}
%     \foreach \ind in {2}
%       \foreach \source/ \dest in {a\ind/b\ind, b\ind/c\ind, c\ind/d\ind, d\ind/a\ind}
%         \path[edge] (\source) -- (\dest);
%     %\foreach \source/ \dest in {c1/g1, d1/h1}
%     %  \path[edge] (\source) -- (\dest);
%     \foreach \source/ \dest in {c2/g2, d2/h2, b3/c3, c3/d3, d3/a3}
%        \path[edgesel] (\source) -- (\dest);
%    \foreach \source/ \dest in {a3/b3, c3/g3, d3/h3}
%       \path[edgerem] (\source) -- (\dest);
%   \end{tikzpicture}
  \caption{Configurations resulting from an application of Step~$3(a)$ and recursive calls on $G, F\cup\{yz\}$, resp. on $G\setminus\{yz\}, F$.}
  \label{step3a}
%\end{floatingfigure}
\end{figure}

In our proof we will therefore analyze all possible configurations which can occur due to applications of Step~$3(a)$ or Step~$3(a')$ and subsequent iterations of Step~1. 
Figure~\ref{step3a} shows the only two configurations due to Step~$3(a)$, Fig.~\ref{6cyclecases} shows {\em all} possible pairwise non-equivalent configurations due to Step~$3(a')$. Speaking more precisely, the figures show {\em fragments} of configurations illustrating the essential parts of graph $G$ and forced edges $F$ after a  recursive call either on $G,F\cup\{yz\}$ or on $G\setminus\{yz\},F$, where $yz$ is the edge chosen  in Step~$3(a)$, resp. Step~$3(a')$.
The configurations will be discussed in detail in the following points. Note that the number of possibilities for cases involving $C(3)$ and $C(4)$ is reduced due to our modification of the algorithm.

%\vspace*{-6mm}
\begin{figure}[htbp]
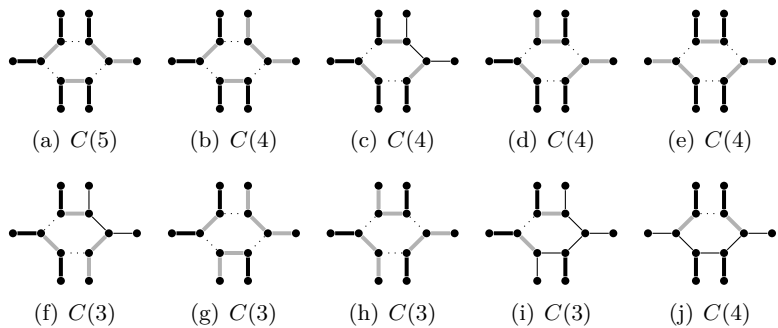

  \centering
  \subfigure[$C(5)$]{
    \begin{tikzpicture}[node distance=0.25]
      \input{6cycle}

      \foreach \source/ \dest in {d/j, e/k, f/l, a/g, b/h}
        \path[edgef] (\source) -- (\dest);
      \foreach \source/ \dest in {a/b, c/d, e/f}
        \path[edgerem] (\source) -- (\dest);
      \foreach \source/ \dest in {b/c, d/e, f/a, c/i}
        \path[edgesel] (\source) -- (\dest);
    \end{tikzpicture}
    \label{6cycle-a}
  }
  %\hspace{0.1cm}
  \subfigure[$C(4)$]{
    \begin{tikzpicture}[node distance=0.25]
      \input{6cycle}

      \foreach \source/ \dest in {d/j, e/k, f/l, a/g}
        \path[edgef] (\source) -- (\dest);
      \foreach \source/ \dest in {a/b, c/d, e/f}
        \path[edgerem] (\source) -- (\dest);
      \foreach \source/ \dest in {b/c, d/e, f/a, c/i, b/h}
        \path[edgesel] (\source) -- (\dest);
    \end{tikzpicture}
    \label{6cycle-b}
  }
  \subfigure[$C(4)$]{
    \begin{tikzpicture}[node distance=0.25]
      \input{6cycle}

      \foreach \source/ \dest in {d/j, e/k, f/l, a/g}
        \path[edgef] (\source) -- (\dest);
      \foreach \source/ \dest in {b/c, c/i, b/h}
        \path[edge] (\source) -- (\dest);
      \foreach \source/ \dest in {d/e, f/a}
        \path[edgerem] (\source) -- (\dest);
      \foreach \source/ \dest in {a/b, c/d, e/f}
        \path[edgesel] (\source) -- (\dest);
    \end{tikzpicture}
    \label{6cycle-c}
  }
  \subfigure[$C(4)$]{
    \begin{tikzpicture}[node distance=0.25]
      \input{6cycle}

      \foreach \source/ \dest in {d/j, e/k, f/l, b/h}
        \path[edgef] (\source) -- (\dest);
      \foreach \source/ \dest in {b/c, d/e, f/a}
        \path[edgerem] (\source) -- (\dest);
      \foreach \source/ \dest in {a/b, c/d, e/f, c/i, a/g}
        \path[edgesel] (\source) -- (\dest);
    \end{tikzpicture}
    \label{6cycle-d}
  }
  \subfigure[$C(4)$]{
    \begin{tikzpicture}[node distance=0.25]
      \input{6cycle}

      \foreach \source/ \dest in {d/j, e/k, b/h, a/g}
        \path[edgef] (\source) -- (\dest);
      \foreach \source/ \dest in {b/c, d/e, f/a}
        \path[edgerem] (\source) -- (\dest);
      \foreach \source/ \dest in {a/b, c/d, e/f, c/i, f/l}
        \path[edgesel] (\source) -- (\dest);
    \end{tikzpicture}
    \label{6cycle-e}
  }

  \subfigure[$C(3)$]{
    \begin{tikzpicture}[node distance=0.25]
      \input{6cycle}

      \foreach \source/ \dest in {a/g, f/l, e/k}
        \path[edgef] (\source) -- (\dest);
      \foreach \source/ \dest in {b/c, b/h, c/i}
        \path[edge] (\source) -- (\dest);
      \foreach \source/ \dest in {d/e, f/a}
        \path[edgerem] (\source) -- (\dest);
      \foreach \source/ \dest in {a/b, c/d, e/f, d/j}
        \path[edgesel] (\source) -- (\dest);
    \end{tikzpicture}
    \label{6cycle-f}
  }
  \subfigure[$C(3)$]{
    \begin{tikzpicture}[node distance=0.25]
      \input{6cycle}

      \foreach \source/ \dest in {a/g, f/l, d/j}
        \path[edgef] (\source) -- (\dest);
      \foreach \source/ \dest in {}
        \path[edge] (\source) -- (\dest);
      \foreach \source/ \dest in {a/b, e/f, c/d}
        \path[edgerem] (\source) -- (\dest);
      \foreach \source/ \dest in {a/f, b/h, b/c, c/i, d/e, e/k}
        \path[edgesel] (\source) -- (\dest);
    \end{tikzpicture}
    \label{6cycle-g}
  }
  \subfigure[$C(3)$]{
    \begin{tikzpicture}[node distance=0.25]
      \input{6cycle}

      \foreach \source/ \dest in {b/h, f/l, d/j}
        \path[edgef] (\source) -- (\dest);
      \foreach \source/ \dest in {}
        \path[edge] (\source) -- (\dest);
      \foreach \source/ \dest in {a/f, b/c, d/e}
        \path[edgerem] (\source) -- (\dest);
      \foreach \source/ \dest in {a/b, e/f, c/d, a/g, c/i, e/k}
        \path[edgesel] (\source) -- (\dest);
    \end{tikzpicture}
    \label{6cycle-h}
  }
  \subfigure[$C(3)$]{
    \begin{tikzpicture}[node distance=0.25]
      \input{6cycle}

      \foreach \source/ \dest in {a/g, f/l, d/j}
        \path[edgef] (\source) -- (\dest);
      \foreach \source/ \dest in {b/h, b/c, c/i, c/d, d/e, e/k}
        \path[edge] (\source) -- (\dest);
      \foreach \source/ \dest in {a/f}
        \path[edgerem] (\source) -- (\dest);
      \foreach \source/ \dest in {a/b, e/f}
        \path[edgesel] (\source) -- (\dest);
    \end{tikzpicture}
    \label{6cycle-i}
  }
  \subfigure[$C(4)$]{
    \begin{tikzpicture}[node distance=0.25]
      \input{6cycle}

      \foreach \source/ \dest in {d/j, e/k, b/h, a/g}
        \path[edgef] (\source) -- (\dest);
      \foreach \source/ \dest in {c/d, e/f, d/e, c/i, f/l}
        \path[edge] (\source) -- (\dest);
      \foreach \source/ \dest in {a/b}
        \path[edgerem] (\source) -- (\dest);
      \foreach \source/ \dest in {b/c, f/a}
        \path[edgesel] (\source) -- (\dest);
    \end{tikzpicture}
    \label{6cycle-j}
  }
  %\hspace{0.5cm}
  \caption{All possible configurations, resulting from an application of Step~$3(a')$ and subsequent iterations of Step~1 on a $C(k)$-cycle such that, in the end, a new edge is added to a $B_2$-cycle $C(i)$ (not seen in the figure), with $3\le i\le k\le 5$.}
  \label{6cyclecases}
\end{figure}
%\vspace*{-6mm}

After performing of Step~$3(a)$ and the subsequent application of Step~1 two configurations can occur, as shown in Fig.~\ref{step3a}. Note, that a $B_2$-cycle $C(i)$, which becomes a $C(j)$-cycle at the end of these steps, is not seen in the figure. We call a 4-cycle {\em isolated}, if all cycle edges are in $G\setminus F$ and all adjacent edges are in $F$. Let us first consider the left configuration of the figure where the edge $yz$ is selected. In consequence another edge attached to the 4-cycle is selected by Step~1. Since the 4-cycle became isolated, two internal edges will be selected by Step~2 at the end of path $P$. 
%We will not count these edges twice because no other branch along $P$ can select these edges again. 
Let us consider the right configuration of Fig.~\ref{step3a} where the edge $yz$ is removed from $G$. In consequence a second attached edge and one cycle edge are removed and three of the cycle edges are selected by iterations of Step~1. These three cycle edges are internal edges. Thus, any branch $Q_\ell$ 
%some of the branches $Q_1,Q_2,\ldots, Q_t$ are 
performed due to Step~$3(a)$ selects at least two internal edges and we are done if such a branch exists in the  sequence $Q_1,Q_2,\ldots, Q_t$.

Suppose then, that none of the branches $Q_1,Q_2,\ldots, Q_t$ is performed due to Step~$3(a)$. Therefore,
%that the configurations for Step~$3(a)$ do not occur and assume applications of branches $Q_\ell$  
all of the branches  have to be performed by an application of Step~$3(a')$. Our aim is to show that if two $B_2$-cycles are active, at least one of the branches selects two internal edges, or two of them select at least one internal edge each, or otherwise if only one $B_2$-cycle is active, one of the branches selects at least one internal edge. We will examine all possible configurations resulting from an  application of Step~$3(a')$ and subsequent iterations of Step~1. The configurations are as shown in Fig.~\ref{6cyclecases}. If one of the branches $Q_\ell$ leads to a configuration \subref{6cycle-a}, \subref{6cycle-b}, \subref{6cycle-d}, \subref{6cycle-e}, \subref{6cycle-g}, or \subref{6cycle-h}, then it is obvious that $Q_\ell$ selects at least two internal edges (cf. also Fig.~\ref{fig:internal:edge}). Hence, we are done. 

Let us assume that a branch $Q_\ell$ yields configuration \subref{6cycle-c}. Obviously there is one internal edge in this configuration. To find the second internal edge, assume one of the four edges which are attached to the 6-cycle in \subref{6cycle-c} has been selected by a $D$-branch. Since this edge is adjacent to two unforced cycle edges before the configuration occurs, it was not considered as internal edge so far. By applying $Q_\ell$ that yields~\subref{6cycle-c}, the edge will be adjacent to another selected edge and can be counted as internal edge now. Hence, we have two internal edges selected for this configuration. Assume all four edges attached to the 6-cycle in \subref{6cycle-c} are selected by $A$-branches. Then the 6-cycle must have been at most $C(2)$ and became the $C(4)$ in \subref{6cycle-c} while activating the $B_2$-cycles by $Q$. Obviously this 6-cycle and a $B_2$-cycle cannot share any selected attached edges, thus two edges selected by $Q$ are attached to the 6-cycle in the 
configuration and the third edge selected by $Q$ can activate at most one $B_2$-cycle, but not two. Hence, in this case, for the single active $B_2$-cycle we assign the internal edge from configuration \subref{6cycle-c}.

The rest of the proof will handle the remaining cases, namely \subref{6cycle-f}, \subref{6cycle-i}, and \subref{6cycle-j}. We will proceed as follows. We fix a $B_2$-cycle activated by $Q$ and denote it by $C$. Then we analyze a sequence of configurations resulting from applications of branches $Q_1, Q_2,\ldots,Q_t$, until $C$ changes to $C(6)$.
% and finally becomes dead. 
Obviously, if one of the configurations \subref{6cycle-a}-\subref{6cycle-e}, \subref{6cycle-g}, \subref{6cycle-h} occurs in this sequence we are done as shown above. Our aim is to prove, that if none of them occurs, meaning that the sequence contains only configurations \subref{6cycle-f}, \subref{6cycle-i}, and \subref{6cycle-j}, then we can assign to $C$ one internal edge and this edge will not be counted twice. To this aim we will consider three cases, namely $C$ transforms from $C(3)$ directly to $C(4)$, or to $C(5)$, or to $C(6)$. We will denote these direct transformations as  $C(3)\to C(j)$, with $4\le j\le 6$, for short.

%Moreover, let $C'$ denote a $C(j)$-cycle, with $3\le j \le 5$ used by a branch realizing transition $C(3)\to C(i)$ which in turn will become dead after the branch is finished. Note that such a $C(j)$-cycle exists since the branch has to be done due to Step~$3(a')$.

Assume first the case $C(3)\to C(4)$ that means $C$ once becomes $C(4)$. Let $C'$ denote a $C(k)$-cycle, with $4\le k \le 5$, used by a branch $Q_{\ell}$ to perform the next transition of $C$ from $C(4)$ to $C(5)$ or to $C(6)$. Obviously, this means that the $C(k)$ will become dead after the branch is finished. Note that such a $C(k)$-cycle exists since the branch has to be done due to Step~$3(a')$. If $C'$ is a $C(5)$-cycle, then we are done, since performing branch $Q_\ell$ on $C(5)$ yields configuration \subref{6cycle-a}. Assume $C'$ is a $C(4)$-cycle. If performing branch $Q_\ell$ on $C(4)$ leads to one of the configurations \subref{6cycle-b}-\subref{6cycle-e}, we are done as well. Thus, the only case we have to consider is that $Q_{\ell}$ leads to configuration \subref{6cycle-j}. In this case it is easy to see that both edges selected by the configuration have to be attached to~$C$. In Fig.~\ref{6cycle-j-application} example cycles $C$ and $C'$  are shown such that the application of $Q_\ell$  results 
in configuration~\subref{6cycle-j}.

%\begin{floatingfigure}[r]{5.2cm}
\begin{figure}[htbp]
   \centering
   \begin{tikzpicture}[node distance=0.7, scale=2]
      \input{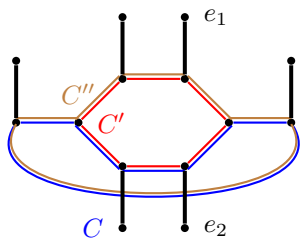}

      \node[vertex] (m) [above=of i] {};
      \node[vertex] (n) [above=of l] {};

%       \foreach \ind in {a,b,c,d,e,f,i,l} {
%         \node[node distance=-0.07] (\ind1) [below=of \ind] {};
%       }

      \node[label=right:$e_1$] at (h) {};
      \node[label=right:$e_2$] at (j) {};
      \node[label=above:$\color{brown} C''$] at (f) {};
      \node[label=right:$\color{red} C'$] at (f) {};
      \node[label=left:$\color{blue} C$] at (k) {};
      %\node[label=right:$e_3$] at (m) {};

      \tikzstyle{edge1} = [line width=0.8pt,-,color=blue]
      \tikzstyle{edge2} = [line width=0.8pt,-,color=red]
      \tikzstyle{edge3} = [line width=0.8pt,-,color=brown]
%       \foreach \source/ \dest in {a/b, b/c, f/a, c/i, l/f}
%         \draw[edge3] (\source) to[bend right=10] (\dest);
%       \draw[edge3] (i) to [out=-60,in=-120] (l);

      \draw[edge1] ($(c)+(0.02,-0.02)$) to ($(d)+(0.02,-0.02)$);
      \draw[edge1] ($(d)+(0,-0.03)$) to ($(e)+(0,-0.03)$);
      \draw[edge1] ($(e)+(-0.02,-0.02)$) to ($(f)+(-0.02,-0.02)$);
      \draw[edge1] (l) to (f);
      \draw[edge1] (c) to (i);
      \draw[edge1] (i) to [out=-60,in=-120] (l);

      \draw[edge3] ($(a)+(0,0.03)$) to ($(b)+(0,0.03)$);
      \draw[edge3] ($(b)+(0.02,0.02)$) to ($(c)+(0.02,0.02)$);
      \draw[edge3] ($(f)+(-0.02,0.02)$) to ($(a)+(-0.02,0.02)$);
      \draw[edge3] ($(l)+(0,0.03)$) to ($(f)+(0,0.03)$);
      \draw[edge3] ($(c)+(0,0.03)$) to ($(i)+(0,0.03)$);
      \draw[edge3] ($(i)+(0,-0.01)$) to [out=-60,in=-120] ($(l)+(0,-0.01)$);

      \foreach \source/ \dest in {d/j, e/k, b/h, a/g, l/n, i/m}
        \path[edgef] (\source) -- (\dest);

      \foreach \source/ \dest in {c/d, e/f, d/e, a/b, b/c, f/a}
        \draw[edge2] (\source) to (\dest);

    \end{tikzpicture}
   \caption{$B_2$- cycle $C$ (blue), 6-cycle $C'$ (red) and $C(4)$-cycle $C''$ (brown). }
   \label{6cycle-j-application}
%\end{floatingfigure}
\end{figure}

Thus, both cycles $C$ and $C'$ are $C(4)$ directly before branch $Q_\ell$ has lead to configuration~\subref{6cycle-j}. But, since $C$ is a $B_2$-cycle, it was $C(3)$  before and made the transition $C(3)\rightarrow C(4)$. Assume all four selected attached edges of $C'$ have been selected by $A$-branches. Then $C'$ was already $C(4)$ when $C$ was $C(3)$ and the transition $C(3)\rightarrow C(4)$ has to be made by a $C(4)$ or $C(5)$ in one of the configurations \subref{6cycle-a}-\subref{6cycle-e}, but not in \subref{6cycle-j}. In this case we are done. Assume next that three selected attached edges of $C'$ have been selected by $A$-branches and one edge by a $D$-branch. Assume that one of the edges shared with $C$, e.\,g. edge $e_2$, was selected by a $D$-branch. Then the other three selected edges attached to $C$ have to be selected by $A$-branches and we have a $C(4)$ cycle $C''$ with four attached edges selected by $A$-branches. Thus, the transition $C(3)\rightarrow C(4)$ of cycle $C$ again has to be due to 
one of the configurations \subref{6cycle-a}-\subref{6cycle-e}. Last we assume that one of the edges which are not shared with $C$, e.\,g. edge $e_1$, was selected by a $D$-branch. Then this edge is neither a $B$-attached edge nor was it counted as internal edge before, since it is adjacent to two unforced edges before the changes in configuration \subref{6cycle-j}. Thus, we can consider this edge to be the internal edge we are looking for. Moreover, due to configuration \subref{6cycle-j} one can only select edges attached to a single $B_2$-cycle, leaving the other $B_2$-cycle unchanged. Therefore we have found one internal edge for a $B_2$-cycle and this edge will not be counted for a second activated $B_2$-cycle, if such a cycle exists.

If the case $C(3)\to C(5)$ occurs, meaning that $C$ once becomes $C(5)$, then only configuration \subref{6cycle-a} can be applied to change $C(5)$ further to $C(6)$ and we are done. 
Thus, the only case which remains to be considered is a direct transition of $C$ from $C(3)$ to $C(6)$ without becoming $C(4)$ or $C(5)$ in between. Recall that an occurrence of the configurations \subref{6cycle-a}-\subref{6cycle-e}, \subref{6cycle-g} and \subref{6cycle-h} selects three internal edges, so that only \subref{6cycle-f}, \subref{6cycle-i} and \subref{6cycle-j} are left to discuss for this change.

\begin{figure}[htbp]
%\begin{wrapfigure}{r}{0pt}
   \centering
    \begin{tikzpicture}[node distance=0.3]
     \input{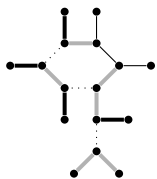}
   
      \foreach \source/ \dest in {a/g, f/l, e/k}
        \path[edgef] (\source) -- (\dest);
      \foreach \source/ \dest in {b/c, b/h, c/i}
        \path[edge] (\source) -- (\dest);
      \foreach \source/ \dest in {d/e, f/a}
        \path[edgerem] (\source) -- (\dest);
      \foreach \source/ \dest in {a/b, c/d, e/f, d/j}
        \path[edgesel] (\source) -- (\dest);
      \node[vertex] (c1) [right=of j] {};
      \node[vertex] (d1) [below=of j] {};
      \node[vertex] (e1) [below left=of d1] {};
      \node[vertex] (f1) [below right=of d1] {};
      \foreach \source/ \dest in {j/c1}
        \path[edgef] (\source) -- (\dest);
      \foreach \source/ \dest in {j/d1}
        \path[edgerem] (\source) -- (\dest);
      \foreach \source/ \dest in {d1/e1, d1/f1}
        \path[edgesel] (\source) -- (\dest);
    \end{tikzpicture}
    \caption{Configuration \subref{6cycle-f} extended by a pattern which forces selection or removal of additional edges during iterations in Step~1.  }
   \label{6cycle-f-extension}
\end{figure}
%\end{wrapfigure}
We can conclude, that occurrence of \subref{6cycle-i} or \subref{6cycle-j} cannot enforce the change $C(3)\rightarrow C(6)$, since in the configurations at most two edges are selected. Due to configuration \subref{6cycle-f} at least one internal edge is selected. On the other hand, three additional edges are selected due to this configuration. Thus, at most one $B_2$-cycle changes from $C(3)$ to $C(6)$ and an internal edge for this cycle is found. 
However, we have to be careful in this case. As explained at the beginning of the proof,  Fig.~\ref{6cyclecases} shows only fragments of configurations illustrating the essential parts of  $G$ and $F$. Thus, in particular some graphs in Fig.~\ref{6cyclecases} do not show necessarily all edges selected  during iterations in Step~1. In case of \subref{6cycle-f}, for example, if an adjacent edge to a newly  selected one is in $F$ then we can obtain a configuration like the one presented in Fig.~\ref{6cycle-f-extension}. Thus, it is possible that configuration \subref{6cycle-f} can be extended by some patterns to select more edges than just three. Some of these edges can be attached to further $B_2$-cycles. Below, we analyse this case in detail. The reader who is not interested in the details can skip it.

%A complete discussion on this matter is given in the appendix.
%%-------------------

%{\bf Further details to the proof of  Lemma~\ref{lemma4b2}.}
%{\bf Further details to the proof.}
\begin{claim}
Suppose a $B_2$-cycle $C$ changes from $C(3)$ directly to $C(6)$ due to a branch yielding configuration {\em \subref{6cycle-f}}. Then the branch selects at least one internal edge that we can assign to $C$. Moreover, if some other $B_2$-cycle is still active then there exists another branch selecting an additional internal edge.
\end{claim}
\begin{proof}[Proof of the Claim]
We have the case that $C$ changes from $C(3)$ to $C(6)$ without becoming $C(4)$ or $C(5)$ in between.
Since we want to count the total number of edges selected in the configuration we have to analyze how the configuration can possibly be extended and thereby select more edges than shown in Fig.~\ref{6cyclecases}\subref{6cycle-f}. There are two patterns which can occur in combination with the configuration, shown in Fig.~\ref{patterns}. The pattern from Fig.~\ref{pattern1} occurs, if an edge selected in the configuration is adjacent to an edge in $F$. Then the third edge incident to the same vertex is removed and two more edges are selected. This pattern can also occur in a chain multiple times. However, if this pattern is found one edge selected by the configuration or by the pattern itself can be considered as internal edge, since one adjacent edge is in $F$ and the other one is removed. Thus, a single occurrence of Pattern~1 lets us find two internal edges for configuration~\subref{6cycle-f} and we are done.

\begin{figure}[htbp]
  \centering
  \subfigure[Pattern~1]{
    \begin{tikzpicture}[node distance=0.35, scale=1]
      \node[vertex] (a) {};
      \node[vertex] (b) [right=of a] {};
      \node[vertex] (c) [above right=of b] {};
      \node[vertex] (d) [below right=of b] {};
      \node[vertex] (e) [below=of d] {};
      \node[vertex] (f) [below right=of d] {};

      \draw[->] (1.2,0) -- (1.5,0);

      \foreach \source/ \dest in {b/c}
        \path[edgef] (\source) -- (\dest);
      \foreach \source/ \dest in {a/b}
        \path[edgesel] (\source) -- (\dest);
      \foreach \source/ \dest in {b/d, d/e, d/f}
        \path[edge] (\source) -- (\dest);
    \end{tikzpicture}
    \hspace{0.2cm}
    \begin{tikzpicture}[node distance=0.35, scale=1]
      \node[vertex] (a) {};
      \node[vertex] (b) [right=of a] {};
      \node[vertex] (c) [above right=of b] {};
      \node[vertex] (d) [below right=of b] {};
      \node[vertex] (e) [below=of d] {};
      \node[vertex] (f) [below right=of d] {};

      \foreach \source/ \dest in {b/c}
        \path[edgef] (\source) -- (\dest);
      \foreach \source/ \dest in {b/d}
        \path[edgerem] (\source) -- (\dest);
      \foreach \source/ \dest in {a/b, d/e, d/f}
        \path[edgesel] (\source) -- (\dest);
    \end{tikzpicture}
    \label{pattern1}
  }
  \hspace{0.5cm}
  \subfigure[Pattern~2]{
    \begin{tikzpicture}[node distance=0.35, scale=1]
      \node[vertex] (a) {};
      \node[vertex] (b) [below right=of a] {};
      \node[vertex] (c) [below=of b] {};
      \node[vertex] (d) [below left=of c] {};
      \node[vertex] (e) [right=of b] {};
      \node[vertex] (f) [right=of c] {};
      \node[vertex] (g) [right=of e] {};
      \node[vertex] (h) [right=of f] {};

      \draw[->] (1.5,-0.4) -- (1.8,-0.4);

      \foreach \source/ \dest in {b/e, c/f}
        \path[edgef] (\source) -- (\dest);
      \foreach \source/ \dest in {b/c}
        \path[edgesel] (\source) -- (\dest);
      \foreach \source/ \dest in {a/b, c/d, e/f, e/g, f/h}
        \path[edge] (\source) -- (\dest);
    \end{tikzpicture}
    \begin{tikzpicture}[node distance=0.35, scale=1]
      \node[vertex] (a) {};
      \node[vertex] (b) [below right=of a] {};
      \node[vertex] (c) [below=of b] {};
      \node[vertex] (d) [below left=of c] {};
      \node[vertex] (e) [right=of b] {};
      \node[vertex] (f) [right=of c] {};
      \node[vertex] (g) [right=of e] {};
      \node[vertex] (h) [right=of f] {};

      \foreach \source/ \dest in {b/e, c/f}
        \path[edgef] (\source) -- (\dest);
      \foreach \source/ \dest in {b/c, e/g, f/h}
        \path[edgesel] (\source) -- (\dest);
      \foreach \source/ \dest in {e/f}
        \path[edgerem] (\source) -- (\dest);
      \foreach \source/ \dest in {a/b, c/d}
        \path[edge] (\source) -- (\dest);
    \end{tikzpicture}
    \label{pattern2}
  }
  \caption{Patterns, which can extend the configurations from Fig.~\ref{6cyclecases}.}
  \label{patterns}
\end{figure}
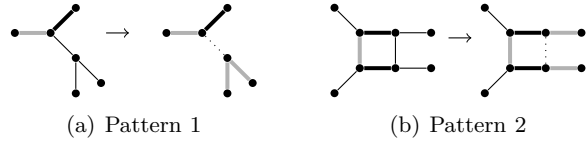

The pattern from Fig.~\ref{pattern2} occurs, if a selected edge is part of a 4-cycle, while the two adjacent cycle edges are in $F$. Due to edge contraction, the three edges will be contracted to a single edge and then form a parallel edge to the fourth cycle edge. The unforced cycle edge will be removed and two further edges selected. The first selected edge can be considered as internal edge. Thus, if this pattern occurs twice we are done. Therefore let us assume, that this pattern occurs only once.

% For the configurations \subref{6cycle-f}, \subref{6cycle-i} and \subref{6cycle-j}, the patterns can only be applied to \subref{6cycle-f}. We can thereby conclude, that \subref{6cycle-i} and \subref{6cycle-j} cannot make the change $C(3)\rightarrow C(6)$, since they select at most two edges in total. Configuration \subref{6cycle-f} selects at least one internal edge. If the first pattern is applied once, a second internal edge is selected and we are done, so that we assume that the first pattern does not occur. Now, the configuration shown in the figure selects three other edges. In addition, the second pattern can be applied once in such a way, that the internal edge counted for the configuration is the same edge as we found in the pattern. Thus, we still have only one internal edge selected for five other edges.

If there is only one $B_2$-cycle active, then one internal edge selected by \subref{6cycle-f} is sufficient. If there are two $B_2$-cycles active, both have to be $C(3)$, since otherwise \subref{6cycle-f} cannot be applied. Since at most five other edges are selected by \subref{6cycle-f}, only one of the two $B_2$-cycles can perform the change $C(3)\rightarrow C(6)$. The other cycle either stays $C(3)$ or becomes $C(4)$ or $C(5)$. If it stays $C(3)$, then one internal edge was selected to change one $B_2$-cycle to $C(6)$ and then only one $B_2$-cycle is left, so that \subref{6cycle-f} would select again one internal edge or one of the other configurations has to be applied as discussed before. If the cycle becomes $C(4)$ or $C(5)$, also the discussion above applies.
\end{proof}
%%-------------------

Summarizing, our analysis shows that for any path $P$ of the backtrack tree the following property holds. Let $P'$ be a maximal subpath of $P$ such that along $P'$ at least one $B_2$-cycle is active. Then some $D$-branch on $P'$ selects at least one internal edge if a single $B_2$-cycle is active on $P'$ and otherwise some $D$-branches along $P'$ select at least two internal edges. We have proven this property analyzing all possible configurations which can occur due to applications of Step~$3(a)$ or Step~$3(a')$ and subsequent iterations of Step~1. From Lemma~\ref{b2b3structure} we know that along $P'$ there are at most two active $B_2$-cycles. Thus, for every subpath $P'$ defined as above, the number of internal edges selected by $D$-branches along $P'$ is greater or equal to the number of $B_2$-cycles which are active in $P'$. Finally, we can conclude that if on $P$ there are totally $b_2$ $B_2$-cycles then at least $b_2$ internal edges are selected on $P$. This follows from the property that every $B_2$-
cycle has to be activated once on path $P$ and that subpaths $P'$ do not overlap. 
\end{proof}

%\begin{proof}(Appendix for proof of Lemma~\ref{lemma4b2})

\begin{proof}[Proof of Lemma~\ref{lemma4b3}]
We proceed analogously as in the proof of Lemma~\ref{b2b3structure} to show that a single $A$-branch can activate at most one $B_3$-cycle. Again, we use the property stated in Lemma~\ref{b3properties} that the only way in which the algorithm makes a $B_3$-cycle active is via an $A$-branch transition $C(2)\rightarrow C(4)$.
Then, the first part of the lemma
%that two attached edges of $C$ have been selected by $D$-branches 
follows easily. Indeed, since $C$ becomes $C(4)$ on the path, from Lemma~\ref{ref:lemma:d:branch} we know that the last two attached edges to $C$ are edges selected by $D$-branches. Next, we show the second part.

From Lemma~\ref{b3properties} we also know that after performing the branch $Q$ which activates $C$, a branch due to Step~$3(a)$ has to be directly followed. Let us denote by  $\tilde{C}$ a 4-cycle affected by Step~$3(a)$ and by $yz$ an edge chosen in this step. Then the algorithm recursively calls on $G, F\cup\{yz\}$ and on $G\setminus\{yz\}, F$ (see Fig.~\ref{step3a:first:call} and \ref{step3a:second:call}). In the first case, we get an isolated 4-cycle for which, in Step~2, the shortest Hamiltonian cycle is found based on a graph representing the isolated 4-cycles. The solution corresponds then to a Hamiltonian cycle, which selects exactly two cycle edges of each isolated 4-cycle. Thus, in this case, we can associate with $C$ two of four edges of $\tilde{C}$. If a recursive call is made on 
$G\setminus\{yz\}, F$, then in consequence three cycle edges of $\tilde{C}$ are selected by iterations of Step~1 and all of them are internal edges.

%the two internal edges we are looking for, are the selected edges of the 4-cycle $\tilde{C}$.
% By Lemma~\ref{b3properties} (which will be proved in the next subsection) a $B_3$-cycle can only be achieved, if one cycle becomes $C(4)$ and the next branch is due to Step~$3(a)$. 
% 
% Directly after $C$ became $C(4)$, Step~$3(a)$ will be applied and a 4-cycle will be isolated. Thus, for every cycle $C$ changing from $C(2)$ to $C(4)$ by an $A$-branch, there is an isolated 4-cycle. In Step~2 of the algorithm, the shortest Hamiltonian cycle is found based on a graph representing the isolated 4-cycles. The solution corresponds then to a Hamiltonian cycle, which selects exactly two cycle edges of each isolated 4-cycle.
\end{proof}

\subsection{An Auxiliary  Result}
\label{subsection:auxiliary}
%\noindent {\bf \large Further Lemmata for Section~\ref{sec:proofs}}\\[2mm]

% Wir verwenden dieses Lemma nur im Beweis eines Lemma!
% \begin{lemma}
% Let $P$ be a single path of the backtrack tree and let $C$ be a $B_2$-cycle. Then the first three attached edges have been selected by $A$-branches and the last three attached edges have been selected by $D$-branches.
% \label{b2properties}
% \end{lemma}
% \begin{proof}
% The cycle $C$ changes on $P$ from $C(0)$ to $C(3)$ and then from $C(3)$ to $C(6)$ (directly or indirectly). Due to Lemma~\ref{onlyd} the last three attached edges are selected by $D$-branches. Since a $B_2$-cycle has more than two attached edges selected by $A$-branches, the first three edges have to be selected by $A$-branches and the Lemma follows.
% \end{proof}

\begin{lemma}
Let $P$ be a single path of the backtrack tree and let $C$ be a $B_3$-cycle on $P$. Then the only way in which the algorithm makes $C$ active on $P$ is via an $A$-branch transition $C(2)\rightarrow C(4)$ which has to be  directly followed by a branch due to Step~$3(a)$.
\label{b3properties}
\end{lemma}
% \begin{lemma}
% Let $C$ be a $C(6)$ along a path in the backtrack tree and $C$ has not been $C(3)$ before. Then either (i) at least four attached edges of $C$ have been selected by $D$-branches or (ii) two attached edges of $C$ have been selected by $D$-branches and two additional edges have been selected due to Step~2.
% \label{atleast3byd}
% \end{lemma}
\begin{proof}
% Consider a cycle $C$ along a path in the backtrack tree, which at some point becomes $C(3)$. Then we know by Lemma~\ref{c3andmore} that all three attached edges, which are not yet selected, have to be selected by $D$-branches. Thus, Lemma~\ref{atleast3byd} is true.
Recall, that if $C$ is a $B_3$-cycle, it becomes $C(6)$ without becoming $C(3)$ in between and at least three attached edges of $C$ are selected by $A$-branches. We will show that in most cases, when transforming a 6-cycle $C$ from $C(0)$ to $C(6)$ (without $C(3)$ in between), at least four attached edges of $C$ are selected by $D$-branches, which means $C$ is a $B_1$-cycle and not a $B_3$ one. Moreover, we will prove that the only exception is the case when initially $C(0)$ is transformed to $C(2)$ by (one or two) $A$-branches and next, a single transition $C(2)\rightarrow C(4)$ via an $A$-branch is made that is immediately followed by a branch due to Step~$3(a)$. Remaining transitions, if needed, are done by $D$-branches (this is a consequence of Lemma~\ref{ref:lemma:d:branch}). This will prove the lemma.

Assume a 6-cycle $C$ transforms along a path $P$ from $C(0)$ to $C(6)$ without becoming $C(3)$ in between. If $C$ becomes also neither $C(4)$ nor $C(5)$, it has a direct change from $C(0)$, $C(1)$ or $C(2)$ to $C(6)$. In all these cases at least four   edges are selected in a single branch, hence, the branch has to be a $D$-branch and $C$ cannot be a $B_3$-cycle.
Thus, we assume that $C$ becomes $C(6)$ with becoming $C(4)$ or $C(5)$ in between, but, of course, without becoming $C(3)$. We will therefore discuss the following direct transitions of~$C$:
\[C(0) \rightarrow C(4),\quad C(1) \rightarrow C(4),\quad C(2) \rightarrow C(4),\]
\[C(0) \rightarrow C(5),\quad C(1) \rightarrow C(5),\quad C(2) \rightarrow C(5).\]
Since every branch which selects more than three edges has to be a $D$-branch, three of these cases indicate immediately that $C$ cannot be a $B_3$-cycle. Thus, only the transitions
\begin{equation}\label{eq:three:transitions}
  \quad C(1) \rightarrow C(4),\quad C(2) \rightarrow C(4),\quad C(2) \rightarrow C(5)
\end{equation}
have to be discussed further. Our aim is to show, that these transitions have to be made by $D$-branches, with the exception as described above.  
%This is equivalent to show that the transitions (except of some cases of $C(2) \rightarrow C(4)$) are not made by an $A$-branch, since a $B$-branch cannot change the degree of another $C(i)$.
Since the case $C(2) \rightarrow C(4)$ is the most involved, we will discussed it as the last one.

Let us assume that the transitions \eqref{eq:three:transitions} were made by $A$-branches. Since Step~$3(a)$ can only result in a $D$-branch and because immediately before starting to branch, $G\setminus F$ contains a live 6-cycle (e.\,g.\ cycle $C$), an edge $zy$ used for the recursive call has to be chosen by Step~$3(a')$. 
Let $C'$ denote the live 6-cycle containing edge $zy$ 
%Consequently, denote by $C'$ a live 6-cycle, an edge of which is chosen in Step~$3(a')$, and  
which in turn will become dead after the branch is finished. Since we assume that an $A$-branch was applied, we can conclude that $C'$ is a $C(j)$-cycle, with $1\le j\le 2$. This follows from the fact shown in the proof of Lemma~\ref{ref:lemma:d:branch} (cf. Fig.~\ref{6cycle-a-branch}), stating that if an $A$-branch affects a $C(j)$-cycle then $j\le 2$. \\[2mm]
\noindent
{\bf Case $\mathbf{C(1) \rightarrow C(4)}$.}  Let us call the branch $Q$. A resulting configuration of the branch is shown in Fig.~\ref{6cycle-a-branch}. Since the transition  changes $C$ from $C(1)$ to  $C(4)$, all edges selected by the branch have to be attached edges of $C$.
Two of the edges %selected by the $A$-branch 
are adjacent to each other and, thus, will be contracted to a single edge by Step~1, before the next branch will be performed. If this edge then forms a triangle, the edge will be removed and cannot be an attached edge of a $C(4)$. Hence these two edges have to be attached to opposite vertices of $C$. Due to symmetry, we can assume the third selected edge to be attached to any of the four other vertices of $C$. The pattern we have concluded to so far is shown in Fig.~\ref{c1toc4}. The naming of the vertices $y$ and $z$ is chosen in a  consistent way with the notation used for the $A$-branch in Fig.~\ref{6cycle-a-branch}. 

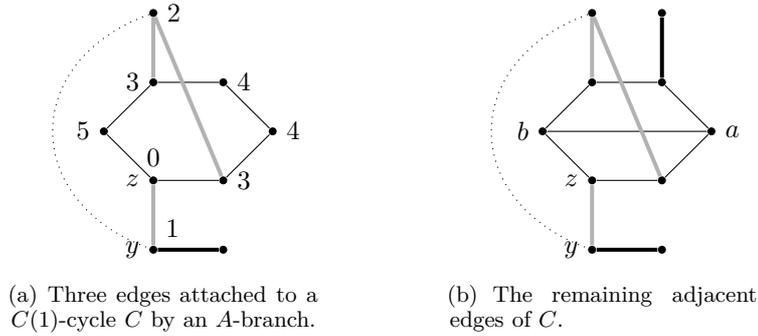
\begin{figure}[htbp]
  \centering
  \subfigure[Three edges attached to a $C(1)$-cycle $C$ by an $A$-branch.]{
    \label{c1toc4}
    \begin{tikzpicture}[node distance=0.8*1, scale=1]
      \node[vertex,label=left:$3$] (a) {};
      \node[vertex,label=right:$4$] (b) [right=of a] {};
      \node[vertex,label=right:$4$] (c) [below right=of b] {};
      \node[vertex,label=right:$3$] (d) [below left=of c] {};
      \node[vertex,label=above:$0$,label=left:$z$] (e) [left=of d] {};
      \node[vertex,label=left:$5$] (f) [above left=of e] {};

      \node[vertex,label=right:$2$] (g) [above=of a] {};
      \node[vertex,label=above right:$1$,label=left:$y$] (h) [below=of e] {};
      \node[vertex] (i) [right=of h] {};

      \foreach \source/ \dest in {a/b, b/c, c/d, d/e, e/f, f/a}
        \path[edge] (\source) -- (\dest);
      \foreach \source/ \dest in {h/i}
        \path[edgef] (\source) -- (\dest);
      \foreach \source/ \dest in {a/g, e/h, g/d}
        \path[edgesel] (\source) -- (\dest);

      \path[edgerem] (g) to [out=200,in=160,min distance=1.8cm] (h);
      %\path[edgesel] (g) to [out=0,in=0,min distance=1.4cm] (d);
    \end{tikzpicture}
  }
  \hspace*{1.5cm}
  \subfigure[The remaining adjacent edges of $C$.]{
    \label{c1toc4:b}
    \begin{tikzpicture}[node distance=0.8*1, scale=1]
      \node[vertex] (a) {};
      \node[vertex] (b) [right=of a] {};
      \node[vertex,label=right:$a$] (c) [below right=of b] {};
      \node[vertex] (d) [below left=of c] {};
      \node[vertex,label=left:$z$] (e) [left=of d] {};
      \node[vertex,label=left:$b$] (f) [above left=of e] {};

      \node[vertex] (g) [above=of a] {};
      \node[vertex,label=left:$y$] (h) [below=of e] {};
      \node[vertex] (i) [right=of h] {};
      \node[vertex] (j) [above=of b] {};

      \foreach \source/ \dest in {a/b, b/c, c/d, d/e, e/f, f/a, c/f}
        \path[edge] (\source) -- (\dest);
      \foreach \source/ \dest in {h/i, b/j}
        \path[edgef] (\source) -- (\dest);
      \foreach \source/ \dest in {a/g, e/h, g/d}
        \path[edgesel] (\source) -- (\dest);

      \path[edgerem] (g) to [out=200,in=160,min distance=1.8cm] (h);
      %\path[edgesel] (g) to [out=0,in=0,min distance=1.4cm] (d);
    \end{tikzpicture}
  }
  \caption{An $A$-branch attaching three forced edges to a $C(1)$-cycle $C$. Figure (a) shows cycle $C$, but with only three adjacent edges that are selected by $A$-branch $Q$. Cycle $C'$ which becomes dead by the branch, is not fully seen in the figure. Figure (b) shows the only possible placement  for the remaining adjacent edges of $C$.}
\end{figure}

Next, our aim is to determine the missing edges adjacent to $C$. We will see that the solution is unique. Starting with the situation as presented in Fig.~\ref{c1toc4} (including the edge shown as a dotted line), we  establish first edges which had to belong to the cycle $C'$, before $Q$ made it dead. Obviously, the chosen edge $zy$ and the adjacent dotted edge were cycle edges of $C'$. The next cycle edge was one of the two selected edges adjacent to the dotted edge. To determine the remaining edges, we proceed as follows. For each candidate vertex $v$, assuming it belongs to $C'$, we determine a distance from $z$ to $v$ on a path along the cycle $C'$ starting from $z$ in direction to $y$. We will see the distance values are unique. Thus, $z$ has distance $0$. Vertex $y$ is the only vertex having distance 1 and the vertex incident to the both selected edges is the only vertex having distance 2. It is easy to verify that its both adjacent vertices must have  distance value 3, meaning that only one of them 
could belong to $C'$.  
The distances for the remaining vertices are 4 or 5. We claim that the unique values are as presented in  Fig.~\ref{c1toc4}. Indeed, the vertices marked with distance 4 cannot have distance 5 since from those vertices no direct edge to $z$ can exist (this could imply degree four of $z$). Next, since at least one vertex has to be of distance 5, the claim follows. Now, the  missing edge of $C'$, let us call it $ab$, have to join a vertex of distance 4 with that one of distance  5. Since $ab$ may not form a triangle with edges of $C$, we get as solution a graph as shown in Fig.~\ref{c1toc4:b}. Thus, the only possible cycle $C'$ which could enable branch $Q$ resulting in configuration in Fig.~\ref{c1toc4}, is the cycle containing the vertices  $z,y,$ vertex of distance 2,  vertex of distance 3, $a,$ and $b$.

In this way we were able to determine the missing edges of $C$. Since, before branch $Q$, the cycle was a $C(1)$-cycle, the last edge has to be the forced edge and it has to be incident with the only vertex left with a degree lower than three.
Thus, the graph in Fig.~\ref{c1toc4:b} is the only configuration which can be obtained after the transition $C(1) \rightarrow C(4)$.

Now, for $C$ to become $C(6)$ from $C(4)$, the edge $ab$ has to be selected by some further branch. Since this edge is not adjacent to any edge in $F$, $F$ is not empty, and there exists a 4-cycle, two vertices of which are adjacent to edges in $F$, only Step~$3(a)$ can choose this edge and add it to $F$. It is also not possible, that $ab$ gets selected in a consequence of some other branches due to Step~1, since its adjacent edges have to stay unforced. To choose $ab$ by Step~$3(a)$, there must be a 4-cycle of unforced edges attached to $ab$. However, there is no such 4-cycle, so that $C$ can never become $C(6)$.
Thus, if a 6-cycle $C$ becames $C(4)$ through an $A$-branch  $C(1) \rightarrow C(4)$, then $C$ cannot be a $B_3$-cycle.\\[2mm]
\noindent
{\bf Case $\mathbf{C(2) \rightarrow C(5)}$.} Due to prioritisation used in Step~$3(a')$, the cycle $C'$ used by the branch $Q$ has to be $C(2)$ (see  Fig.~\ref{6cycle-a-branch-in-c2--6cycle-a-branch-for-c5}(a)). 
Note that the second selected attached edge of $C'$ has to be incident to the vertex $c$. Otherwise the vertex $b$ would not be free and we would not have an $A$-branch, since the other configuration of the branch would select one more edge.

\begin{figure}[htbp]
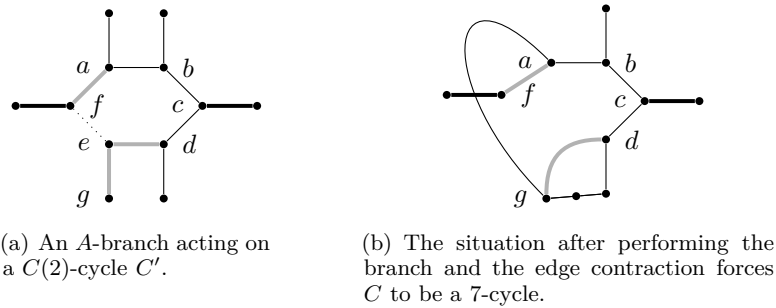

  \centering
  \subfigure[An $A$-branch acting on a $C(2)$-cycle $C'$.]{
    \label{6cycle-a-branch-in-c2}
    \begin{tikzpicture}[node distance=0.6*1, scale=1]
      \input{6cycle}

      \input{a-in-c2-labels}

      \foreach \source/ \dest in {a/g, b/h, d/j, e/k, a/b, b/c, c/d}
        \path[edge] (\source) -- (\dest);
      \foreach \source/ \dest in {f/l, c/i}
        \path[edgef] (\source) -- (\dest);
      \foreach \source/ \dest in {e/f}
        \path[edgerem] (\source) -- (\dest);
      \foreach \source/ \dest in {f/a, d/e, e/k}
        \path[edgesel] (\source) -- (\dest);
    \end{tikzpicture}
  }
   \hspace*{1.0cm}
  \subfigure[The situation after performing the branch and the edge contraction forces $C$ to be a 7-cycle.]{
    \hspace*{0.7cm}
    \label{6cycle-a-branch-for-c5}
    \begin{tikzpicture}[node distance=0.6*1, scale=1]
      \node[vertex] (a) {};
      \newcommand{\ind}{}
      \node[vertex] (b\ind) [right=of a\ind] {};
      \node[vertex] (c\ind) [below right=of b\ind] {};
      \node[vertex] (d\ind) [below left=of c\ind] {};
      \node[] (e\ind) [left=of d\ind] {};
      \node[vertex] (f\ind) [above left=of e\ind] {};

      \node[vertex] (g\ind) [above=of a\ind] {};
      \node[vertex] (h\ind) [above=of b\ind] {};
      \node[vertex] (i\ind) [right=of c\ind] {};
      \node[vertex] (j\ind) [below=of d\ind] {};
      \node[vertex] (k\ind) [below=of e\ind] {};
      \node[vertex] (l\ind) [left=of f\ind] {};

      \node[fill=white, minimum size=4pt] at (g) {}; % TODO this is not nice

     \node[label=left:$a$] at (a) {};
     \node[label=right:$b$] at (b) {};
     \node[label=left:$c$] at (c) {};
     \node[label=right:$d$] at (d) {};
 %    \node[label=left:$e$] at (e) {};
     \node[label=right:$f$] at (f) {};
     \node[label=left:$g$] at (k) {};

      \draw (j) -- (k) node[vertex,pos=0.5](m){};

      \foreach \source/ \dest in {b/h, d/j, a/b, b/c, c/d, j/m, m/k}
        \path[edge] (\source) -- (\dest);
      \foreach \source/ \dest in {f/l, c/i}
        \path[edgef] (\source) -- (\dest);
 %     \foreach \source/ \dest in {e/f}
 %       \path[edgerem] (\source) -- (\dest);
      \foreach \source/ \dest in {f/a}
        \path[edgesel] (\source) -- (\dest);

      \path[edge] (a) to [out=135,in=135,min distance=2cm] (k);
      \path[edgesel] (d) to [out=180,in=90,min distance=0.5cm] (k);
    \end{tikzpicture}
    \hspace*{0.7cm}
  }
  \caption{An $A$-branch acting on $C(2)$. Figure (a) shows the situation after performing the branch. % that the two selected edges of the $C(2)$-cycle have to be of distance three.
  Figure (b) shows the situation  after edges $de$ and $eg$ get contracted in Step~1. 
 }
  \label{6cycle-a-branch-in-c2--6cycle-a-branch-for-c5}
\end{figure}

For changing $C$ from $C(2)$ to $C(5)$ all three edges selected by the $A$-branch $Q$ have to be attached edges of $C$. Therefore, the edges $ab$ and $cd$ have to be cycle edges of $C$. Since $c$ is incident to an edge in $F$, $bc$ is also a cycle edge of $C$. Consider the path along the cycle $C$ starting from $a$ over $b$ and $c$ to $d$. The path has to visit $g$ and then complete the 6-cycle by returning to $a$. Since $de$ and $eg$ get contracted, the path from $d$ to $g$ has to have at least three edges as shown in Fig.~\ref{6cycle-a-branch-in-c2--6cycle-a-branch-for-c5}(b). Otherwise, a triangle contraction would be performed on $C$. Thus, we have a path length of three from $a$ to $d$ and a path length of at least three from $d$ to $g$. With at least one more edge to get from $g$ back to $a$ the cycle $C$ has a length of at least seven, which is a contradiction to $C$ being a 6-cycle. Therefore, we conclude that the transition $C(2) \rightarrow C(5)$ cannot be performed by an $A$-branch.

\noindent
{\bf Case $\mathbf{C(2) \rightarrow C(4)}$.} Let us call the branch $Q$. 
%Similarly as in the previous case, due to prioritisation used in Step~$3(a')$, the cycle $C'$ used by $Q$, has to be $C(2)$. 
Obviously, after the branch cycle $C$ becomes $C(4)$. Then it can happen that either $C$ is the only $C(4)$-cycle at that moment or some other cycle becomes $C(4)$, too. In the analysis below we will consider these two subcases separately.\\[2mm]
\noindent
{\bf Subcase: one $\mathbf{C(4)}$-cycle after branch $\mathbf{Q}$.} 
Assume $C$ is the only cycle becoming $C(4)$ by $Q$. Then, to be a $B_3$-cycle, the next branch directly after $Q$ has to be done due to Step~$3(a)$. If this is not the case, then after $Q$, Step~$3(a')$ is performed. Since $C$ is the only $C(4)$-cycle, and because there exists no $C(5)$ cycle (this fact can be seen easily), a branch due to Step~$3(a')$ makes $C$ dead and, thus, $C$ cannot be $B_3$.

%We can conclude that, for the current subcase, if directly after $Q$, Step~$3(a)$ does not follow, then $C$ cannot become $B_3$.  
By Fig.~\ref{4:by:a:branches} we provide an example, showing that it is possible to activate a cycle $C$ by an $A$-branch transition $C(2)\rightarrow C(4)$ followed by Step~$3(a)$. The first transition from $C(0)$ to $C(2)$ has been done by an $A$-branch and then $A$-branch $Q$ performs transition $C(2)\to C(4)$. Next, using a branch due to Step~$3(a)$ the last two attached edges of $C$ are selected, making $C$ a $B_3$-cycle.  

%four attached edges of a cycle $C$ can be selected by $A$-branches, while $C$ performs a transition $C(2) \rightarrow C(4)$ by an $A$-branch $Q$.% Our aim is to show, that for each $C(6)$ cycle of this type, at least two additional edges are selected by Step~2 and thereby not by an $A$- or a $B$-branch.

% Directly after $C$ became $C(4)$, Step~$3(a)$ will be applied and a 4-cycle will be isolated. Thus, for every cycle $C$ changing from $C(2)$ to $C(4)$ by an $A$-branch, there is an isolated 4-cycle. In Step~2 of the algorithm, the shortest Hamiltonian cycle is found based on a graph representing the isolated 4-cycles. The solution corresponds then to a Hamiltonian cycle, which selects exactly two cycle edges of each isolated 4-cycle.

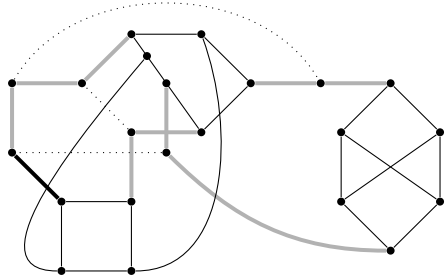
\begin{figure}[htbp]
  \centering
%  \subfigure[attached $A$-branch]{
    \begin{tikzpicture}[node distance=0.8*1, scale=1]
      \node[vertex] (a) {};
      \node[vertex] (b) [right=of a] {};
      \node[vertex] (c) [below right=of b] {};
      \node[vertex] (d) [below left=of c] {};
      \node[vertex] (e) [left=of d] {};
      \node[vertex] (f) [above left=of e] {};
      \draw (a) -- (d) node[vertex,pos=0.2](g){} node[vertex,pos=0.5](h){};
      %\node[vertex] (i) [above=of b] {};
      \node[vertex] (j) [right=of c] {};
      \node[vertex] (k) [left=of f] {};
      \node[vertex] (l) [below=of h] {};
      \node[vertex] (m) [below=of k] {};
      \node[vertex] (o) [below=of e] {};
      \node[vertex] (n) [left=of o] {};
      \node[vertex] (p) [below=of o] {};
      \node[vertex] (q) [left=of p] {};
      \node[vertex] (r) [right=of j] {};
      \node[vertex] (s) [below right=of r] {};
      \node[vertex] (t) [below=of s] {};
      \node[vertex] (u) [below left=of t] {};
      \node[vertex] (w) [below left=of r] {};
      \node[vertex] (v) [below=of w] {};
      %\node[vertex] (x) [below left=of g] {};

      \foreach \source/ \dest in {a/b, b/c, c/d, n/o, o/p, p/q, q/n, r/s, s/t, s/v, r/w, w/t, w/v, t/u, v/u}
        \path[edge] (\source) -- (\dest);
      \foreach \source/ \dest in {m/n}
        \path[edgef] (\source) -- (\dest);
      \foreach \source/ \dest in {m/k, l/h, k/f, f/a, d/e, e/o, c/j, j/r}
        \path[edgesel] (\source) -- (\dest);
      \foreach \source/ \dest in {l/m, e/f}
        \path[edgerem] (\source) -- (\dest);

      \path[edgerem] (k) to [out=60,in=120] (j);
      \path[edgesel] (l) to [out=-45,in=180] (u);
      \path[edge] (b) to [out=-70,in=0] (p);
      \path[edge] (g) to [out=-130,in=180] (q);
    \end{tikzpicture}
%  }
  \caption{A $B_3$-cycle activated by an $A$-branch performing the transition $C(2)\to C(4)$.}
  \label{4:by:a:branches}
\end{figure}

\noindent
{\bf Subcase: two  $\mathbf{C(4)}$-cycles after branch $\mathbf{Q}$.} 
Let us consider the case of two 6-cycles $C$ and $\hat{C}$ becoming $C(4)$ by $Q$. The $A$-branch $Q$ selects at most three attached edges to these 6-cycles. Since $C$ and $\hat{C}$ both are live after branch $Q$, any vertex which belongs to a cycle and is incident to an attached edge selected by $Q$ is also incident to two unforced cycle edges. In branch $Q$, as  shown in Fig.~\ref{6cycle-a-branch}, only one of the two vertices of each selected edge fulfills this condition, so that each selected edge can be attached only at one end to a cycle and not at both ends. For the transitions of $C$ and $\hat{C}$, at least four attached edges have to be selected by $Q$. Hence, the attached edges are not disjoint. The cycles $C$ and $\hat{C}$ have to share at least one attached edge which gets selected by $Q$ and its two adjacent cycle edges. We consider different cases depending on how many attached edges are shared between $C$ and $\hat{C}$.

\begin{figure}[htbp]
  \centering
  \begin{tikzpicture}[node distance=0.6*1, scale=1]
    \node[vertex] (a) {};
    \node[vertex] (b) [right=of a] {};
    \node[vertex] (c) [below right=of b] {};
    \node[vertex] (d) [below left=of c] {};
    \node[vertex] (e) [left=of d] {};
    \node[vertex] (f) [below=of a] {};

    \node[vertex] (g) [left=of e] {};
    \node[vertex] (h) [above left=of g] {};
    \node[vertex] (i) [above right=of h] {};

    \node[] (j) [above=of b] {};
    \node[] (k) [right=of c] {};
    \node[] (l) [below=of d] {};
    \node[] (m) [left=of f] {};
    \node[] (n) [below=of g] {};
    \node[] (o) [left=of h] {};
    \node[] (p) [above=of i] {};

    \foreach \source/ \dest in {a/b, b/c, c/d, d/e, e/f, f/a, e/g, g/h, h/i, i/a}
      \path[edge] (\source) -- (\dest);
    \foreach \source/ \dest in {b/j, c/k, d/l, f/m, g/n, h/o, i/p}
      \path[edgef] (\source) -- (\dest);
  \end{tikzpicture}
  \caption{Case 1: One shared attached edge. The figure shows the situation immediately after $Q$ has been completed.}
  \label{one:shared:edge}
\end{figure}
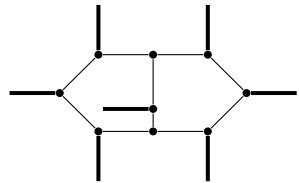

In the first case, $C$ and $\hat{C}$ share only one attached edge and, thus, its both adjacent edges are common cycle edges of $C$ and $\hat{C}$. Figure~\ref{one:shared:edge} shows the cycles immediately after $Q$ has been completed. Both cycles $C$ and $\hat{C}$ have to be at least $C(4)$, so that all non-cycle edges have to be selected. We observe that removing or selecting an arbitrary edge of $\hat{C}$ will also select or remove an edge of $C$. Therefore, $C$ becomes dead and not $C(6)$.

\begin{figure}[htbp]
  \centering
  \subfigure[Shared {\em not neighboring} edges.]
   {
%     \begin{tikzpicture}[node distance=0.3*1, scale=1]
%       \node[] (a) {};
%       \node[vertex] (b) [right=of a] {};
%       \node[] (c) [right=of b] {};
%       \node[vertex] (d) [right=of c] {};
%       \node[] (e) [right=of d] {};
%       \node[vertex] (f) [right=of e] {};
%       \node[] (g) [right=of f] {};
% 
%       \node[vertex] (h) [above=of c] {};
%       \node[vertex] (i) [above=of e] {};
%       \node[] (j) [above=of i] {};
%  
%      \node[vertex] (k) [below=of c] {};
%       \node[vertex] (l) [below=of e] {};
%       \node[] (m) [below=of l] {};
% 
%       \foreach \source/ \dest in {a/b, d/e, f/g, b/h, h/i, i/f, f/l, l/k, k/b, h/d, d/k}
%         \path[edge] (\source) -- (\dest);
%       \foreach \source/ \dest in {i/j, l/m}
%         \path[edgef] (\source) -- (\dest);
%     \end{tikzpicture}
   \hspace*{1.2cm}
    \begin{tikzpicture}[node distance=0.6*1, scale=1]
      \node[vertex] (a) {};
      \node[vertex] (b) [right=of a] {};
      \node[vertex] (c) [right=of b] {};
      \node[vertex] (d) [below=of c] {};
      \node[vertex] (e) [left=of d] {};
      \node[vertex] (f) [left=of e] {};

      \node[] (g) [above=of b] {};
      \node[] (h) [below=of e] {};

      \foreach \source/ \dest in {a/b, b/c, c/d, d/e, e/f, f/a, c/f, a/d}
        \path[edge] (\source) -- (\dest);
      \foreach \source/ \dest in {b/g, e/h}
        \path[edgef] (\source) -- (\dest);
    \end{tikzpicture}
    \hspace{1.2cm}
    \label{two:shared:edges:not:neighboring}
  }
  \hspace*{.5cm}
  \subfigure[Shared {\em neighboring} edges.]{ \hspace*{0.5cm}
    \begin{tikzpicture}[node distance=0.6*1, scale=1]
      \input{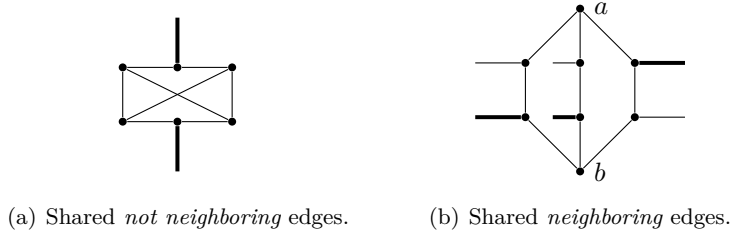}

      \foreach \source/ \dest in {a/b, e/f, d/g, c/j, h/n, f/l}
        \path[edge] (\source) -- (\dest);
    \end{tikzpicture}
    \hspace*{0.5cm}
    \label{two:shared:edges:neighboring}
  }
  \caption{Case 2: Two shared attached edges. Figures (a) and (b) show situations right before $Q$ starts.}
  \label{two:shared:edges}
\end{figure}

In the second case, $C$ and $\hat{C}$ share two attached edges, as shown in Fig.~\ref{two:shared:edges}. If these two shared attached edges are not neighboring, the cycles share four cycle edges and we have a situation as shown in Fig.~\ref{two:shared:edges:not:neighboring}. Note, that right before $Q$ starts, both shared attached edges have do be forced. Obviously, these two 6-cycles cannot both become $C(4)$. If the shared attached edges are neighboring, we have the situation as shown in Fig.~\ref{two:shared:edges:neighboring}. Now, observe that since there is no $C(i)$, with $i>2$, right before branch $Q$ starts, on each path from $a$ to $b$ of length 3 at most one edge can be selected. 
On the other hand, after $Q$ has been completed, two edges on each of these paths have to be selected in order to get two $C(4)$'s. An $A$-branch selects exactly three edges, so that right before $Q$ starts, on each $a-b$ path exactly one edge has to be selected. Furthermore, if all three edges selected by an $A$-branch are attached to $C$ and $\hat{C}$, as required in this case, two attached edges selected by $Q$ need to have a distance of three. Thus, the forced edges right before $Q$ starts, have 
to be attached to $C$ and $\hat{C}$ as shown in Fig.~\ref{two:shared:edges:neighboring}.

\begin{figure}[htbp]
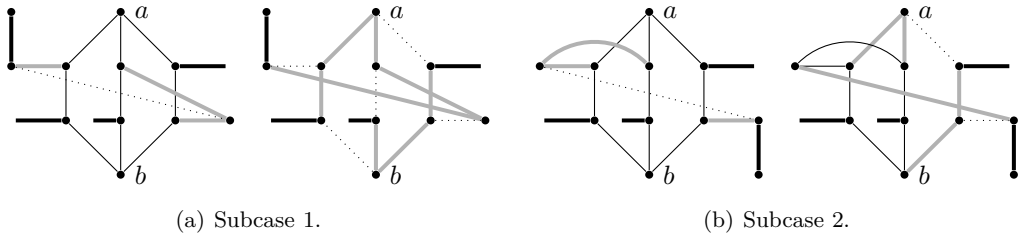

  \centering
  \subfigure[Subcase 1.]{
    \begin{tikzpicture}[node distance=0.6*1, scale=1]
      \input{two-shared-att}

      \node[vertex] (j) [right=of c] {};
      \node[vertex] (n) [left=of h] {};
      \node[vertex] (o) [above=of n] {};

      \foreach \source/ \dest in {a/b, e/f, d/g}
        \path[edge] (\source) -- (\dest);
      \foreach \source/ \dest in {n/o}
        \path[edgef] (\source) -- (\dest);
      \foreach \source/ \dest in {h/n, j/f, c/j}
        \path[edgesel] (\source) -- (\dest);
      \foreach \source/ \dest in {n/j}
        \path[edgerem] (\source) -- (\dest);
    \end{tikzpicture}
    \begin{tikzpicture}[node distance=0.6*1, scale=1]
      \input{two-shared-att}

      \node[vertex] (j) [right=of c] {};
      \node[vertex] (n) [left=of h] {};
      \node[vertex] (o) [above=of n] {};

      \foreach \source/ \dest in {n/o}
        \path[edgef] (\source) -- (\dest);
      \foreach \source/ \dest in {n/j, g/h, h/a, c/d, d/e, f/a, f/j, b/c}
        \path[edgesel] (\source) -- (\dest);
      \foreach \source/ \dest in {h/n, d/g, e/f, a/b, c/j}
        \path[edgerem] (\source) -- (\dest);
    \end{tikzpicture}
    \label{two:shared:edges:n1}
  }
  \subfigure[Subcase 2.]{
    \begin{tikzpicture}[node distance=0.6*1, scale=1]
      \input{two-shared-att}

      \node[vertex] (j) [right=of c] {};
      \node[vertex] (n) [left=of h] {};
      \node[vertex] (o) [below=of j] {};

      \foreach \source/ \dest in {a/b, e/f, d/g}
        \path[edge] (\source) -- (\dest);
      \foreach \source/ \dest in {j/o}
        \path[edgef] (\source) -- (\dest);
      \foreach \source/ \dest in {h/n, c/j}
        \path[edgesel] (\source) -- (\dest);
      \foreach \source/ \dest in {n/j}
        \path[edgerem] (\source) -- (\dest);

      \path[edgesel] (n) to [out=45,in=135] (f);%,min distance=1.8cm
    \end{tikzpicture}
    \begin{tikzpicture}[node distance=0.6*1, scale=1]
      \input{two-shared-att}

      \node[vertex] (j) [right=of c] {};
      \node[vertex] (n) [left=of h] {};
      \node[vertex] (o) [below=of j] {};

      \foreach \source/ \dest in {e/f, d/g, h/n}
        \path[edge] (\source) -- (\dest);
      \foreach \source/ \dest in {j/o}
        \path[edgef] (\source) -- (\dest);
      \foreach \source/ \dest in {h/a, f/a, b/c, c/d, j/n}
        \path[edgesel] (\source) -- (\dest);
      \foreach \source/ \dest in {a/b, c/j}
        \path[edgerem] (\source) -- (\dest);

      \path[edge] (n) to [out=45,in=135] (f);%,min distance=1.8cm
    \end{tikzpicture}
    \label{two:shared:edges:n2}
  }
  \caption{Subcases of Case~2: Two possibilities to attach the three edges selected by $Q$.}
  \label{two:shared:edges:n}
\end{figure}

Recall that, by assumption, $Q$ is an $A$-branch acting on a $C(2)$-cycle $C'$ (Fig.~\ref{6cycle-a-branch}).
Up to symmetries, there are two possibilities to attach the three edges selected  by $Q$ to $C$ and $\hat{C}$, as shown in Fig.~\ref{two:shared:edges:n}. We can observe, that in no case branch $Q$ can be an $A$-branch. %is no longer an $A$-branch. 
% Instead we have
% \begin{equation*}
% \begin{split}
% T(s-3,a,b,f-4)+T(s-8,a,b,f-6),\\
% T(s-3,a,b,f-4)+T(s-5,a,b,f-6)
% \end{split}
% \end{equation*}
% for Fig.~\ref{two:shared:edges:n1} and \ref{two:shared:edges:n2}, respectively.

\begin{figure}[htbp]
  \centering
  \begin{tikzpicture}[node distance=0.6*1, scale=1]
    \node[vertex] (a) {};
    \node[vertex,label=above left:$a$] (f) [below=of a] {};
    \node[vertex] (e) [below=of f] {};
    \node[vertex,label=below left:$b$] (d) [below=of e] {};
    \node[vertex] (c) [below=of d] {};

    \node[vertex, node distance=1.5] (b) [right=of e] {};
    \node[vertex, node distance=1.5] (g) [left=of e] {};

    \node[] (h) [right=of b] {};
    \node[] (i) [right=of d] {};
    \node[] (j) [right=of e] {};
    \node[] (k) [right=of f] {};
    \node[] (l) [left=of g] {};

    \foreach \source/ \dest in {a/b, b/c, c/d, d/e, e/f, a/f, c/g, g/a}
      \path[edge] (\source) -- (\dest);
    \foreach \source/ \dest in {b/h, d/i, e/j, f/k, g/l}
      \path[edgef] (\source) -- (\dest);
  \end{tikzpicture}
  \caption{Case~3: Three shared attached edges.}
  \label{three:shared:edges}
\end{figure}
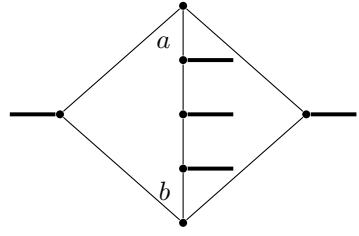

In the third case, $C$ and $\hat{C}$ have three shared attached edges. Assume that these shared edges are not neighboring. Then $C$ is entered and exited by $\hat{C}$ at least two times at disjoint vertices, i.\,e.\ $C$ has at least four vertices for entering and leaving. Additionally, $C$ has three vertices with the shared attached edges. This is a contradiction, since $C$ has exactly 6 vertices.
Thus, the three shared edges have to be neighboring. Figure~\ref{three:shared:edges} shows $C$ and $\hat{C}$ immediately after $Q$ has been completed. Then the graph contains a 4-cycle with opposite attached edges that are selected. Due to Step~1(i), the edges $a$ and $b$ are selected and $C$ cannot become $C(6)$.

For the last case, assume that $C$ and $\hat{C}$ share more than three attached edges. Then they share at least six vertices and thus have to be identical. This contradicts to the assumption, that $C$ will become $C(6)$ when $\hat{C}$ becomes dead.

Summarizing, we found that it is not possible to create two $C(4)$ by an $A$-branch $Q$, where one $C(4)$ becomes $C(6)$ in a subsequent step. Therefore $Q$ has to be a $D$-branch.
\end{proof}

\section{Comments on Work by Iwama and Nakashima}
%Some Issues Concerning the Paper \ref{rec:eq:iwama:nakashima})
\label{section:comments:on:IN}

In \cite{IwamaN07} Iwama and Nakashima modify Eppstein's algorithm in a slightly different way than we do. To analyse the run-time of the modified algorithm the authors give the following recurrence relation 
\begin{equation}\label{eq:rec:iwama}
T(n,a,b,f)\le\text{max}\left.\begin{cases}
              2T(n-3,a-1,b,f-4)\\
              2T(n-3,a,b-1,f)\\
              T(n-5,a,b,f-2)+T(n-2,a,b,f-2)\\
              2T(n-4,a,b,f).
             \end{cases}\right.
\end{equation}
Next, they let $T(n,a,b,f)= 2^{\frac{n+\frac{1}{2}(a+2b)+f/8}{4}}$, verify that this function satisfies relation (\ref{eq:rec:iwama}) and, using the inequality $3a+6b\le n$, they derive the upper bound on the run-time as follows: $T(n,a,b,n)=2^{\frac{n+\frac{1}{2}(a+2b)+n/8}{4}}\le 2^{\frac{n+\frac{1}{2}n/3+n/8}{4}}=2^{31n/96}$. 

The authors show, among others, that $$T(n,a,b,f) < T(n-5,a,b,f-2) + T(n-2,a,b,f-2).$$ This expresses a rather strange property that the run-time of the algorithm for the problem could be less than the sum of the run-times for the both subproblems. On the other hand the authors do not prove that the presented upper bound is valid. Note that e.\,g.\ function $T(n,a,b,f)=1$ also satisfies relation (\ref{eq:rec:iwama}) and the base case $T(0,0,0,0)=1$. Thus, following the reasoning in \cite{IwamaN07} one could derive $T(n,a,b,n)=1$ as an upper bound on the run-time of the algorithm.

% To prove that the function $T$ satisfies the recurrence above they show, among others, that $T(n-5,a,b,f-2) + T(n-2,a,b,f-2) = 2^{\frac{(n-5)+\frac{1}{2}(a+2b)+(f-2)/8}{4}} + 2^{\frac{(n-2)+\frac{1}{2}(a+2b)+(f-2)/8}{4}}>T(n,a,b,f)$. This inequality expresses a rather strange property that the upper bound on the run-time of the algorithm for the problem, with parameters $n,a,b,f$, is smaller than the sum of upper bounds on the run-times for the both subproblems, with parameters $n-5,a,b,f-2$ and  $n-2,a,b,f-2$. Hence, in particular, one cannot conclude that the presented upper bound on the run-time of the algorithm is valid. Note that if one lets as solution  e.\,g.\ $T(n,a,b,f)=n+8$ then it satisfies the conditions in \cite{IwamaN07}, too. Thus, following the reasoning in \cite{IwamaN07} one could derive $T(n,a,b,n)=n+8$ as an upper bound on the run-time of the algorithm. 

A second error  concerns the key lemma of \cite{IwamaN07} (Lemma~1). It says that if a 6-cycle $Q$ becomes $C(6)$, then at least three attached edges of $Q$ have been selected by $D$-branches.
Recently, we have discovered that the lemma is false. As a counterexample, we were able to construct  a (small) cubic graph such that for some path $P$ of the branching tree there exists a 6-cycle $Q$ becoming $C(6)$ on $P$ having only two attached edges selected by $D$-branches (for this counterexample see Appendix). To be correct, the lemma should be reformulated as follows: if a 6-cycle $Q$ becomes $C(6)$ then at least {\em two} attached edges of $Q$ have been selected by $D$-branches (and the bound two is best possible). After reformulating the key lemma and then solving a proper recurrence one could conclude the upper bound $\caO(1.257^n)$.

\section{Conclusions}

In this paper we have provided a new upper bound $\mathcal{O}(1.2553^n)$ for TSP in cubic graphs which consequently also applies for the Hamiltonian cycle problem. We have shown that the exact algorithm of Eppstein with some minor modifications, yields this result. An interesting open problem is to further improve this bound. One could try e.\,g.\ to find a new algorithm and prove a better asymptotic time complexity than $\mathcal{O}(1.2553^n)$. On the other hand, we believe that the worst-case time complexity of Eppstein's algorithm is much smaller than the current upper bound. %(see discussion on our experimental results below). 
Hence, another approach to resolve the problem would be to improve the analysis of the algorithm.

Our upper bound follows from the main technical contribution of this paper that estimates the number of worst-case branches, so called $A$- and $B$-branches, along any path of the backtrack tree. However, constructing backtrack trees containing a worst-case  path we have observed that they result in very 'unbalanced' trees: out of the worst-case paths, the remaining paths are short. Thus, one direction in improving our upper bound could be to improve the estimation of the worst-case number of $A$- and $B$-branches in the whole tree and incorporate this information in an analysis of the worst-case size of the backtrack tree.

Our experimental analysis \cite{Schuster2012} has confirmed that Eppstein's algorithm with our modification is easily to implement and that the algorithm has good performance. Additionally we have shown a gap between our upper bound on the tree size and actual sizes for graphs up to 112 vertices. This could indicate that the worst-case complexity of Eppstein's algorithm is much smaller than  $\mathcal{O}(1.2553^n)$. 

We analyzed the number of branches made by the implementation for random graphs with up to 112 vertices. For each graph size, five random graphs were generated. Since the algorithm is randomized, we run the implementation three times for each input graph. Figure~\ref{total-branches} shows the results of these calls using a logarithmic scale. The maximum number of branches used for the 5 random graphs of size $n$ is shown by the red line. The orange line shows the average of the calls. We further used $(3,g)$-cages, i.\,e.\  3-regular graphs of girth $g$ of minimum order (\cite{Tutte47}, see \cite{Exoo2011} for a survey), with girth values from 3 to 11 as input for our implementation. The branches used in these calls are indicated by the brown line in the figure. The graphs have up to 112 vertices for a girth of 11. Finally, a family of graphs  presented by Eppstein in \cite{Eppstein07} was examined. The family, indexed by their number of vertices $n$, is constructed such that any $n$-vertex graph of the family 
has $2^{n/3}$ Hamiltonian cycles. The results for running the implementation on these graphs is indicated by the purple line for up to 114 vertices.

\begin{figure}[htbp]
 \centering
  %{\bf A new  figure here}

\begin{tikzpicture}[domain=0:118,xscale=0.1,yscale=0.15]
%\draw[very thin,color=gray] (-0.1,-1.1) grid (3.9,3.9);
\draw[->] (-0.2,0) -- (115.2,0) node[right] {$n$};
\draw[->] (0,-0.2) -- (0,28.2) node[above] {\#Branches};
\foreach \x in {20,40,60,80,100}
  \draw (\x,0.5) -- (\x,-0.5) node[anchor=north] {$\x$};
\foreach \y in {10,20}
  \draw (0.5,\y) -- (-0.5,\y) node[anchor=east] {$2^{\y}$};
\draw[color=black,dashed] plot[domain=0:90] (\x,{\x*log2(1.2553)}) node[left] {$R(n,n,n/4,n/7,n) = 1.2553^n$};
\draw[color=green,dashed] plot (\x,{\x*log2(1.15)}) node[right] {$f(n) = 1.15^n$};
\draw[color=blue] plot[mark=x] coordinates {
(1,0.000) (2,0.000) (3,0.000) (4,1.000) (5,1.000) (6,1.585) (7,2.000) (8,2.322) (9,2.585) (10,3.170) (11,3.459) (12,3.700) (13,4.170) (14,4.459) (15,4.755) (16,5.170) (17,5.459) (18,5.755) (19,6.087) (20,6.459) (21,6.700) (22,7.087) (23,7.426) (24,7.700) (25,8.066) (26,8.426) (27,8.700) (28,9.066) (29,9.358) (30,9.700) (31,9.977) (32,10.358) (33,10.700) (34,10.977) (35,11.322) (36,11.700) (37,11.954) (38,12.322) (39,12.687) (40,12.954) (41,13.295) (42,13.687) (43,13.954) (44,14.313) (45,14.600) (46,14.954) (47,15.219) (48,15.600) (49,15.931) (50,16.219) (51,16.585) (52,16.931) (53,17.209) (54,17.585) (55,17.901) (56,18.209) (57,18.581) (58,18.901) (59,19.209) (60,19.581) (61,19.833) (62,20.209) (63,20.512) (64,20.839) (65,21.190) (66,21.512) (67,21.833) (68,22.190) (69,22.459) (70,22.833) (71,23.127) (72,23.459) (73,23.825) (74,24.127) (75,24.447) (76,24.825) (77,25.090) (78,25.447) (79,25.783) (80,26.090) (81,26.426) (82,26.716) (83,27.087) (84,27.426) (85,27.714) (86,28.087) (87,28.361) (88,28.714) (89,29.
071) (90,29.361)% (91,29.697) (92,30.071) (93,30.341) (94,30.697) (95,31.026) (96,31.341) (97,31.655) (98,32.026) (99,32.338) (100,32.655)
} node[right] {$T(n,n,n/4,n/7,n)$};
\draw[color=orange] plot[mark=x] coordinates {
(4,0.000) (8,-0.322) (12,1.263) (16,0.951) (20,3.397) (24,3.553) (28,5.290) (32,6.377) (36,7.379) (40,7.471) (44,8.506) (48,9.299) (52,9.841) (56,11.330) (60,11.179) (64,12.811) (68,12.906) (72,13.951) (76,14.210) (80,15.935) (84,16.712) (88,16.630) (92,17.894) (96,18.873) (100,19.706) (104,20.547) (108,21.082) (112,21.737)
} node[right] {rand.\ avg.};
\draw[color=red] plot[mark=x] coordinates {
(4,0.000) (8,1.000) (12,2.322) (16,2.807) (20,4.524) (24,4.392) (28,5.858) (32,7.190) (36,8.103) (40,8.570) (44,9.331) (48,9.892) (52,10.220) (56,11.997) (60,12.318) (64,13.735) (68,13.699) (72,14.683) (76,15.210) (80,16.506) (84,17.005) (88,17.168) (92,18.512) (96,19.385) (100,20.003) (104,21.361) (108,21.683) (112,22.544)
} node[above left] {rand.\ max.};
% \draw[color=brown] plot[mark=x] coordinates {(10,10)};
\draw[color=brown] plot[mark=x] coordinates {
(4, 0) (6, 1.0) (10, 2.321928094887362) (14, 4.08746284125034) (24, 5.672425341971496) (30, 7.851749041416057) (58, 12.679919878518419) (70, 16.146290228456845) (112, 24.00581272475441)
} node[above] {cages};
\draw[color=purple] plot[mark=x] coordinates {
(6,1.000) (12,2.322) (18,3.459) (24,4.524) (30,5.555) (36,6.570) (42,7.577) (48,8.581) (54,9.583) (60,10.584) (66,11.584) (72,12.585) (78,13.585) (84,14.585) (90,15.585) (96,16.585)
 (102,17.585) (108,18.585) (114,19.585)} node[below right] {HamCycles}; % (120,20.585)
% \fill[brown] (30,7.80735492206) circle (0.5);
% \fill[brown] (70,16.1277928582) circle (0.5);
\end{tikzpicture}

\caption{Total number of branches (in logarithmic scale) in backtrack trees for $n$-vertex random graphs, cages, and graphs of $2^n/3$ Hamiltonian cycles.}
\label{total-branches}
\end{figure}
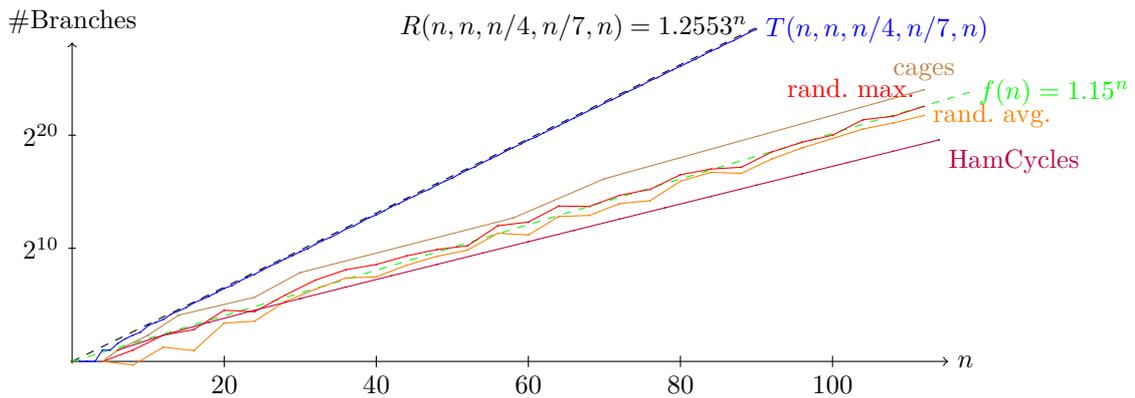

The blue line gives the results for the recurrence function $T(n,s,x,y,f)$ for the start condition deduced by the analysis. The analytical upper bound, given by the function $R$ in our analysis, is shown by the dashed black line. We can observe that the upper bound is a good approximation for $T$.

The number of branches needed to process a cage is higher than the corresponding number for random graphs. However, the ratio seems to be bounded by a constant factor. The same observation can be made for the comparison of the maximum case and the average case for several random graphs of the same size. The worst-case of the tested random graphs stays within a constantly bounded ratio to the average. If we compare the growth of the number of branches with our upper bound of $\mathcal{O}(1.2553^n)$, the experimental results indicate a much better bound of $\mathcal{O}(1.15^n)$, at least for most of the input graphs. Although the graph class corresponding to the purple line has a large number of Hamiltonian cycles, the algorithm needs only a small number of branches and even less branches than some random graph.

%\vspace*{-2mm}

\bibliographystyle{splncs}

\bibliography{literature}

\newpage

\appendix
\section{Eppstein's Algorithm}

%\begin{table} 
%  \caption{Eppstein's Algorithm}
%  \label{algtab}
  \begin{enumerate}%\vspace*{-5mm}
    \item Repeat the following steps until one of the steps returns or none of them applies:
    \begin{enumerate}
      \item[$(a)$] If $G$ contains a vertex with degree zero or one, return {\em None}.
      \item[$(b)$] If $F$ consists of a Hamiltonian cycle, return the cost of this cycle.
      \item[$(c)$] If $F$ contains a non-Hamiltonian cycle, return {\em None}.
      \item[$(d)$] If $F$ contains three edges meeting at a vertex, return {\em None}.
      \item[$(e)$] If $G$ contains two parallel edges, at least one of which is not in $F$, and $G$ has more than two vertices, then remove from $G$ an unforced edge of the two, choosing the one with larger costs if both are unforced. 
      \item[$(f)$] If $G$ contains a self-loop which is not in $F$, and $G$ has more than one vertex, remove the self-loop from $G$.
      \item[$(g)$] If $G$ contains a vertex with degree two, add its incident edges to $F$.
      \item[$(h)$] If $F$ contains exactly two edges meeting at some vertex, remove from $G$ that vertex and any other edge incident to it; replace the two edges by a single forced edge connecting their other two endpoints, having as its cost the sum of the costs of the two replaced edges' costs.
      \item[$(i)$] If $G$ contains a triangle $xyz$, then for each non-triangle edge $e$ incident to a triangle vertex, increase the cost of $e$ by the cost of the opposite triangle edge. Also, if the triangle edge opposite $e$ belongs to $F$, add $e$ to $F$. Remove from $G$ the three triangle edges, and contract the three triangle vertices into a single supervertex.
      \item[$(j)$] If $G$ contains a cycle of four unforced edges, two opposite vertices of which are each incident to a forced edge outside the cycle, then add to $F$ all non-cycle edges that are incident to a vertex of the cycle.
    \end{enumerate}
    \item If $G\setminus F$ forms a collection of disjoint 4-cycles, perform the following steps.
    \begin{enumerate}
      \item[$(a)$] For each 4-cycle $C_i$ in $G\setminus F$, let $H_i$ consist of two opposite edges of $C_i$, chosen so that the cost of $H_i$ is less than or equal to the cost of $C_i\setminus H_i$.
      \item[$(b)$] Let $H=\cup_i H_i$. Then $F\cup H$ is a degree-two spanning subgraph of $G$, but may not be connected.
      \item[$(c)$] Form a graph $G'=(V',E')$, where the vertices of $V'$ consist of the connected components of $F\cup H$. For each set $H_i$ that contains edges from two different components $K_j$ and $K_k$, draw an edge in $E'$ between the corresponding two vertices, with cost equal to the difference between the costs of $C_i$ and of $H_i$.
      \item[$(d)$] Compute the minimum spanning tree of $(G',E')$.
      \item[$(e)$] Return the sum of the costs of $F\cup H$ and of the minimum spanning tree.
    \end{enumerate}
    \item Choose an edge $yz$ according to the following cases:
    \begin{enumerate}
      \item[$(a)$] If $G\setminus F$ contains a 4-cycle, two vertices of which are adjacent to edges in $F$, let $y$ be one of the other two vertices of the cycle and let $yz$ be an edge of $G\setminus F$ that does not belong to the cycle.
%       \item[$(a')$] \emph{If there is no such 4-cycle and if $G\setminus F$ contains a live 6-cycle with a vertex $y$ which has a neighboring edge in $F$ (that is not a cycle edge but an attached one), let $z$ be one of $y$'s neighboring vertices (on the cycle). If two or more such live 6-cycles exist, then select a 6-cycle such that most attached edges are already selected. If there is more than one such edge $yz$ in the 6-cycle, choose $yz$ so, that $z$ also has a neighboring edge in $F$.}
      \item[$(b)$] If there is no such 4-cycle, but $F$ is nonempty, let $xy$ be any edge in $F$ and $yz$ be any adjacent edge in $G\setminus F$.
      \item[$(c)$] If $F$ is empty, let $yz$ be any edge in $G$.
    \end{enumerate}
    \item Call the algorithm recursively on $G$, $F\cup\{yz\}$.
    \item Call the algorithm recursively on $G\setminus\{yz\}$, $F$.
    \item Return the minimum of the set of at most two numbers returned by the two recursive calls.
  \end{enumerate}
%\end{table}

\newpage

\section{A Counterexample to Lemma~1 in \cite{IwamaN07}}
%\noindent {\bf \large A Counterexample to Lemma~1 in \cite{IwamaN07}  }\\[2mm]
\begin{figure}[htbp]
  \centering
  \subfigure[Arbitrary edge selected due to Step~3$(b2)$]{
    \begin{tikzpicture}[node distance=0.7*1, scale=1]
      \node[vertex] (a) {};
      \node[vertex] (b) [right=of a] {};
      \node[vertex] (c) [below right=of b] {};
      \node[vertex] (d) [below left=of c] {};
      \node[vertex] (e) [left=of d] {};
      \node[vertex] (f) [above left=of e] {};
      \draw (a) -- (d) node[vertex,pos=0.2](g){} node[vertex,pos=0.5](h){};
      %\node[vertex] (i) [above=of b] {};
      \node[vertex] (j) [right=of c] {};
      \node[vertex] (k) [left=of f] {};
      \node[vertex] (l) [below=of h] {};
      \node[vertex] (m) [below=of k] {};
      \node[vertex] (o) [below=of e] {};
      \node[vertex] (n) [left=of o] {};
      \node[vertex] (p) [below=of o] {};
      \node[vertex] (q) [left=of p] {};
      \node[vertex] (r) [right=of j] {};
      \node[vertex] (s) [below right=of r] {};
      \node[vertex] (t) [below=of s] {};
      \node[vertex] (u) [below left=of t] {};
      \node[vertex] (w) [below left=of r] {};
      \node[vertex] (v) [below=of w] {};
      %\node[vertex] (x) [below left=of g] {};

      \foreach \source/ \dest in {a/b, b/c, c/d, n/o, o/p, p/q, q/n, r/s, s/t, s/v, r/w, w/t, w/v, t/u, v/u, m/k, l/h, k/f, f/a, d/e, e/o, c/j, j/r, l/m, e/f}
        \path[edge] (\source) -- (\dest);
      \foreach \source/ \dest in {m/n}
        \path[edgesel] (\source) -- (\dest);

      \path[edge] (k) to [out=60,in=120] (j);
      \path[edge] (l) to [out=-45,in=180] (u);
      \path[edge] (b) to [out=-70,in=0] (p);
      \path[edge] (g) to [out=-130,in=180] (q);
    \end{tikzpicture}
  }\hspace*{10mm}
  \subfigure[First $A$-branch]{
    \begin{tikzpicture}[node distance=0.7*1, scale=1]
      \node[vertex] (a) {};
      \node[vertex] (b) [right=of a] {};
      \node[vertex] (c) [below right=of b] {};
      \node[vertex] (d) [below left=of c] {};
      \node[vertex] (e) [left=of d] {};
      \node[vertex] (f) [above left=of e] {};
      \draw (a) -- (d) node[vertex,pos=0.2](g){} node[vertex,pos=0.5](h){};
      %\node[vertex] (i) [above=of b] {};
      \node[vertex] (j) [right=of c] {};
      \node[vertex] (k) [left=of f] {};
      \node[vertex] (l) [below=of h] {};
      \node[vertex] (m) [below=of k] {};
      \node[vertex] (o) [below=of e] {};
      \node[vertex] (n) [left=of o] {};
      \node[vertex] (p) [below=of o] {};
      \node[vertex] (q) [left=of p] {};
      \node[vertex] (r) [right=of j] {};
      \node[vertex] (s) [below right=of r] {};
      \node[vertex] (t) [below=of s] {};
      \node[vertex] (u) [below left=of t] {};
      \node[vertex] (w) [below left=of r] {};
      \node[vertex] (v) [below=of w] {};
      %\node[vertex] (x) [below left=of g] {};

      \foreach \source/ \dest in {a/b, b/c, c/d, n/o, o/p, p/q, q/n, r/s, s/t, s/v, r/w, w/t, w/v, t/u, v/u, e/f, k/f, f/a, d/e, e/o, c/j, j/r}
        \path[edge] (\source) -- (\dest);
      \foreach \source/ \dest in {m/n}
        \path[edgef] (\source) -- (\dest);
      \foreach \source/ \dest in {m/k, l/h}
        \path[edgesel] (\source) -- (\dest);
      \foreach \source/ \dest in {l/m}
        \path[edgerem] (\source) -- (\dest);

      \path[edge] (k) to [out=60,in=120] (j);
      \path[edgesel] (l) to [out=-45,in=180] (u);
      \path[edge] (b) to [out=-70,in=0] (p);
      \path[edge] (g) to [out=-130,in=180] (q);
    \end{tikzpicture}
  }
  \subfigure[Edge contraction and second $A$-branch]{
    \begin{tikzpicture}[node distance=0.7*1, scale=1]
      \node[vertex] (a) {};
      \node[vertex] (b) [right=of a] {};
      \node[vertex] (c) [below right=of b] {};
      \node[vertex] (d) [below left=of c] {};
      \node[vertex] (e) [left=of d] {};
      \node[vertex] (f) [above left=of e] {};
      \draw (a) -- (d) node[vertex,pos=0.2](g){} node[vertex,pos=0.5](h){};
      %\node[vertex] (i) [above=of b] {};
      \node[vertex] (j) [right=of c] {};
      \node[vertex] (k) [left=of f] {};
      %\node[vertex] (l) [below=of h] {};
      %\node[vertex] (m) [below=of k] {};
      \node[vertex] (o) [below=of e] {};
      \node[vertex] (n) [left=of o] {};
      \node[vertex] (p) [below=of o] {};
      \node[vertex] (q) [left=of p] {};
      \node[vertex] (r) [right=of j] {};
      \node[vertex] (s) [below right=of r] {};
      \node[vertex] (t) [below=of s] {};
      \node[vertex] (u) [below left=of t] {};
      \node[vertex] (w) [below left=of r] {};
      \node[vertex] (v) [below=of w] {};
      %\node[vertex] (x) [below left=of g] {};

      \foreach \source/ \dest in {a/b, b/c, c/d, n/o, o/p, p/q, q/n, r/s, s/t, s/v, r/w, w/t, w/v, t/u, v/u, e/f, f/a, d/e, e/o}
        \path[edge] (\source) -- (\dest);
      \foreach \source/ \dest in {}
        \path[edgef] (\source) -- (\dest);
      \foreach \source/ \dest in {k/f, c/j, j/r}
        \path[edgesel] (\source) -- (\dest);
      \foreach \source/ \dest in {}
        \path[edgerem] (\source) -- (\dest);

      \path[edgerem] (k) to [out=60,in=120] (j);
      \path[edge] (b) to [out=-70,in=0] (p);
      \path[edge] (g) to [out=-130,in=180] (q);
      \path[edgef] (h) to [out=-90,in=180] (u);
      \path[edgef] (k) to [out=-90,in=135] (n);
    \end{tikzpicture}
  }\hspace*{10mm}
  \subfigure[Edge contraction and third $A$-branch]{
    \begin{tikzpicture}[node distance=0.7*1, scale=1]
      \node[vertex] (a) {};
      \node[vertex] (b) [right=of a] {};
      \node[vertex] (c) [below right=of b] {};
      \node[vertex] (d) [below left=of c] {};
      \node[vertex] (e) [left=of d] {};
      \node[vertex] (f) [above left=of e] {};
      \draw (a) -- (d) node[vertex,pos=0.2](g){} node[vertex,pos=0.5](h){};
      %\node[vertex] (i) [above=of b] {};
      %\node[vertex] (j) [right=of c] {};
      %\node[vertex] (k) [left=of f] {};
      %\node[vertex] (l) [below=of h] {};
      %\node[vertex] (m) [below=of k] {};
      \node[vertex] (o) [below=of e] {};
      \node[vertex] (n) [left=of o] {};
      \node[vertex] (p) [below=of o] {};
      \node[vertex] (q) [left=of p] {};
      \node[vertex] (r) [right=of j] {};
      \node[vertex] (s) [below right=of r] {};
      \node[vertex] (t) [below=of s] {};
      \node[vertex] (u) [below left=of t] {};
      \node[vertex] (w) [below left=of r] {};
      \node[vertex] (v) [below=of w] {};
      %\node[vertex] (x) [below left=of g] {};

      \foreach \source/ \dest in {a/b, b/c, c/d, n/o, o/p, p/q, q/n, r/s, s/t, s/v, r/w, w/t, w/v, t/u, v/u}
        \path[edge] (\source) -- (\dest);
      \foreach \source/ \dest in {c/r}
        \path[edgef] (\source) -- (\dest);
      \foreach \source/ \dest in {f/a, d/e, e/o}
        \path[edgesel] (\source) -- (\dest);
      \foreach \source/ \dest in {e/f}
        \path[edgerem] (\source) -- (\dest);

      \path[edge] (b) to [out=-70,in=0] (p);
      \path[edge] (g) to [out=-130,in=180] (q);
      \path[edgef] (h) to [out=-90,in=180] (u);
      \path[edgef] (f) to [out=-135,in=90] (n);
    \end{tikzpicture}
  }
  \subfigure[Edge contraction and Step~3(a)]{
    \begin{tikzpicture}[node distance=0.7*1, scale=1]
      \node[vertex] (a) {};
      \node[vertex] (b) [right=of a] {};
      \node[vertex] (c) [below right=of b] {};
      \node[vertex] (d) [below left=of c] {};
      %\node[vertex] (e) [left=of d] {};
      %\node[vertex] (f) [above left=of e] {};
      \draw (a) -- (d) node[vertex,pos=0.2](g){} node[vertex,pos=0.5](h){};
      %\node[vertex] (i) [above=of b] {};
      %\node[vertex] (j) [right=of c] {};
      %\node[vertex] (k) [left=of f] {};
      %\node[vertex] (l) [below=of h] {};
      %\node[vertex] (m) [below=of k] {};
      \node[vertex] (o) [below=of e] {};
      \node[vertex] (n) [left=of o] {};
      \node[vertex] (p) [below=of o] {};
      \node[vertex] (q) [left=of p] {};
      \node[vertex] (r) [right=of j] {};
      \node[vertex] (s) [below right=of r] {};
      \node[vertex] (t) [below=of s] {};
      \node[vertex] (u) [below left=of t] {};
      \node[vertex] (w) [below left=of r] {};
      \node[vertex] (v) [below=of w] {};
      %\node[vertex] (x) [below left=of g] {};

      \foreach \source/ \dest in {a/b, b/c, c/d, n/o, o/p, p/q, q/n, r/s, s/t, s/v, r/w, w/t, w/v, t/u, v/u}
        \path[edge] (\source) -- (\dest);
      \foreach \source/ \dest in {c/r, d/o}
        \path[edgef] (\source) -- (\dest);
      \foreach \source/ \dest in {}
        \path[edgesel] (\source) -- (\dest);
      \foreach \source/ \dest in {}
        \path[edgerem] (\source) -- (\dest);

      \path[edgesel] (b) to [out=-70,in=0] (p);
      \path[edgesel] (g) to [out=-130,in=180] (q);
      \path[edgef] (h) to [out=-90,in=180] (u);
      \path[edgef] (a) to [out=-135,in=90] (n);
    \end{tikzpicture}
  }
  \caption{The input graph is shown in Fig.~(a). Figures (b)-(e) show resulting configurations, i.\,e.\ $G$ and $F$, along a path of the backtrack tree performed by the algorithm of Iwama and Nakashima presented in  \cite{IwamaN07}. The $C(6)$-cycle seen in Fig.~(e) has four attached edges selected by $A$-branches.}
\end{figure}
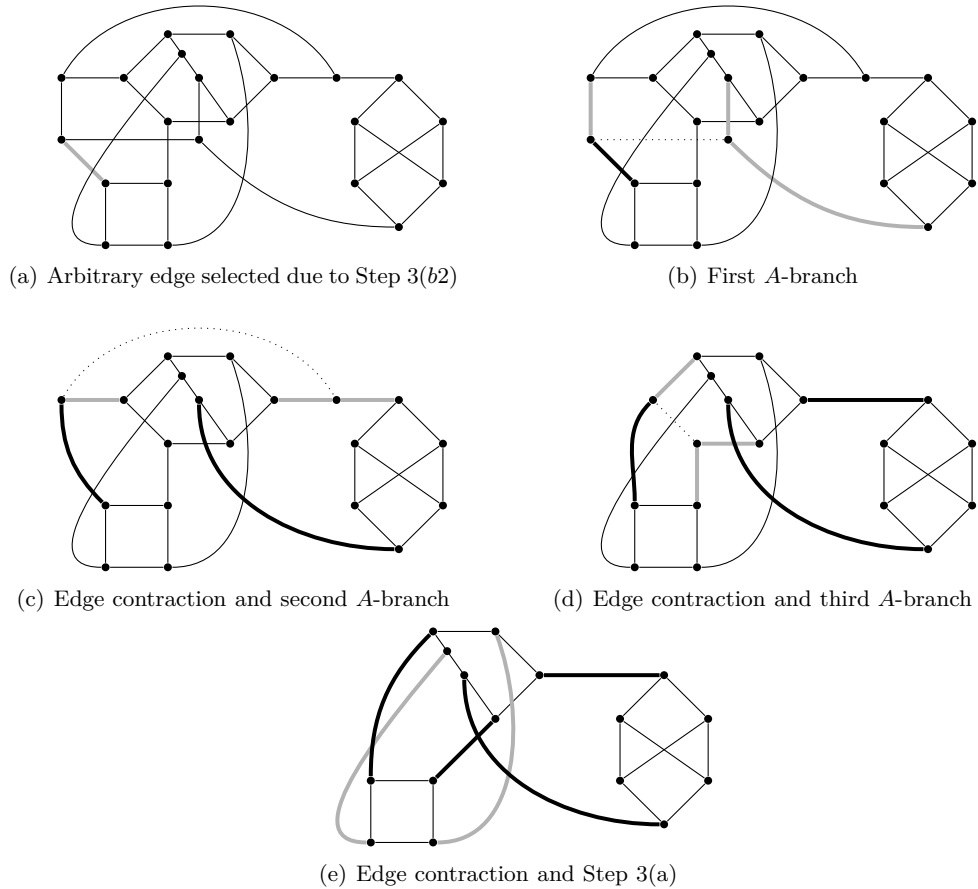

\end{document}